\documentclass[onecolumn,longbibliography,nofootinbib, superscriptaddress,10pt,aps,prx]{revtex4-2}

\usepackage[dvips]{graphicx} 
\usepackage{amsfonts}
\usepackage{amssymb}
\usepackage{amscd}
\usepackage{amsmath}    
\usepackage{amsthm}
\usepackage{appendix}
\usepackage{enumerate}
\usepackage{epsfig}
\usepackage{subfigure}
\usepackage{subfloat}
\usepackage{xcolor}			
\usepackage{amsthm}
\usepackage{physics}
\usepackage[most]{tcolorbox}
\usepackage{tabularx}
\usepackage{float}
\usepackage{appendix}
\usepackage{makecell}
\usepackage{dsfont}
\usepackage[colorlinks, linkcolor=red, anchorcolor=blue, citecolor=green]{hyperref}
\usepackage{MnSymbol}
\usepackage{tikz}
\usetikzlibrary{quantikz2}
\usetikzlibrary{matrix}
\graphicspath{{./figure/}}

\newtheorem{theorem}{Theorem}
\newtheorem{lemma}{Lemma}

\newtheorem{definition}{Definition}

\newcommand{\cT}{\mathcal{T}}
\newcommand{\cF}{\mathcal{F}}
\newcommand{\cS}{\mathcal{S}}
\newcommand{\ve}{\varepsilon}
\newcommand{\CNOT}{\mathrm{CNOT}}
\DeclareMathOperator{\diag}{diag}
\DeclareMathOperator{\argmin}{argmin}

\newtcolorbox[auto counter]{mybox}[2][]{
	enhanced,
	breakable,
	colback=blue!5!white,
	colframe=blue!75!black,
	fonttitle=\bfseries,
	title=Box \thetcbcounter: #2,#1
}

\begin{document}

\title{Advantage Distillation for Quantum Key Distribution}

\author{Zhenyu Du}
\author{Guoding Liu}
\author{Xingjian Zhang}
\author{Xiongfeng Ma}
\email{xma@tsinghua.edu.cn}
\affiliation{Center for Quantum Information, Institute for Interdisciplinary Information Sciences, Tsinghua University, Beijing, 100084 China}

\begin{abstract}
Quantum key distribution promises information-theoretically secure communication, with data post-processing playing a vital role in extracting secure keys from raw data. While hardware advancements have significantly improved practical implementations, optimizing post-processing techniques offers a cost-effective avenue to enhance performance. Advantage distillation, which extends beyond standard information reconciliation and privacy amplification, has proven instrumental in various post-processing methods. However, the optimal post-processing remains an open question. Therefore, it is important to develop a comprehensive framework to encapsulate and enhance these existing methods. In this work, we propose an advantage distillation framework for quantum key distribution, generalizing and unifying existing key distillation protocols. Inspired by entanglement distillation, our framework not only integrates current techniques but also improves upon them. Notably, by employing classical linear codes, we achieve higher key rates, particularly in scenarios where one-time pad encryption is not used for post-processing. Our approach provides insights into existing protocols and offers a systematic way for further enhancements in quantum key distribution.
\end{abstract}

\maketitle

\section{Introduction}
Quantum key distribution (QKD) is the most practical field in quantum information science, with seminal protocols like BB84 \cite{bennett1984quantum} and E91 \cite{Ekert_1991_quantum} laying its foundation. QKD enables the establishment of secret keys between remote parties, conventionally referred to as Alice and Bob, with the assurance of information-theoretic security rooted in the principles of quantum mechanics \cite{Lo1999Unconditional, Shor_simple_2000}. While real-world imperfections in state preparation, transmission, and measurement pose inherent challenges, the development of ingenious techniques like the decoy-state method \cite{hwang2003decoy,Lo2005decoy,wang2005decoy} and cutting-edge measurement-device-independent protocols \cite{Lo_measurement_2012, Lucamarini_2018, Ma_phase_2018, Zeng_mode_2022} has significantly pushed the limits of practical implementations \cite{RevModPhys.92.025002}. In parallel, hardware advancements have driven the evolution of QKD systems, exemplified by satellite-based QKD, which harnesses the unique benefits of space-based platforms to overcome terrestrial limitations in establishing secure long-distance communication links \cite{Liao_2017, Liao2018Satellite}. Additionally, the integration of QKD with quantum repeaters and other quantum communication technologies within quantum networks has paved the way for secure communication among multiple parties over extended distances \cite{Peev_2009, Sasaki_2011, Mohsen2012network, chen2021integrated, Chen2021Implementation, Zhu_field_2024}. These remarkable strides collectively signify a paradigm shift towards achieving quantum-secured communication on a global scale, promising unprecedented security levels for sensitive information exchange.

The QKD process can be dissected into two distinct phases. In the quantum phase, quantum states are prepared and transmitted through channels, often assumed to be susceptible to adversarial control by Eve. Following potential tampering by Eve, these states are measured by detection devices, yielding raw data. Subsequently, the classical post-processing phase commences, where Alice and Bob execute specific procedures to distill secure key bits from this raw data. This involves data sifting, which is dictated by the chosen QKD protocol, followed by preprocessing, information reconciliation, and privacy amplification. The flowchart of QKD, as illustrated in Fig.~\ref{fig:QKDflowchart}, provides a visual representation of this process.

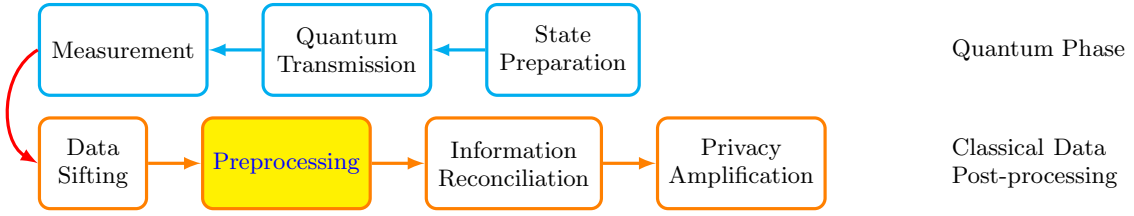
\begin{figure}[htbp!]
	\begin{tikzpicture} [
		rounded corners,
		>=latex,
		Qbox/.style={draw=cyan,rectangle,very thick,minimum height=1.2cm,text width=2cm,align=center,text=black},
		Cbox/.style={draw=orange,rectangle,very thick,minimum height=1.2cm,text width=2cm,align=center,text=black}
		]	
		\node[Qbox,text width=1.8cm] (S) {State Preparation};
		\node[Qbox,anchor=east] (T) at ($(S.west)+(-.7,0)$) {Quantum Transmission};
		\node[Qbox,anchor=east] (M) at ($(T.west)+(-.7,0)$) {Measurement};
		\node[Cbox,text width=1.2cm,anchor=west] (Si) at ($(M.west)+(0,-1.5)$) {Data Sifting};
		\node[Cbox,anchor=west,fill=yellow, fill opacity=0.2,text opacity=1,text=blue] (P) at ($(Si.east)+(.7,0)$) {Preprocessing};
		\node[Cbox,text width=2.1cm,anchor=west] (I) at ($(P.east)+(.7,0)$) {Information Reconciliation};
		\node[Cbox,anchor=west] (Pa) at ($(I.east)+(.7,0)$) {Privacy Amplification};
		\node[text width=3cm, anchor=west] (QP) at ($(S.east) + (+4,0)$) {Quantum Phase};
		\node[text width=3cm, anchor=west] (CP) at ($(QP.west)+ (0, -1.5)$) {Classical Data Post-processing};
		
		\draw[->,very thick,cyan] (S.west) -- (T.east);
		\draw[->,very thick,cyan] (T.west) -- (M.east);
		\draw[->,very thick,red] (M.west) to [bend right=60]  (Si.west);
		\draw[->,very thick,orange] (Si.east) -- (P.west);
		\draw[->,very thick,orange] (P.east) -- (I.west);
		\draw[->,very thick,orange] (I.east) -- (Pa.west);
	\end{tikzpicture}
	\caption{Flowchart for QKD. The process consists of two phases. Top: The quantum phase involves state preparation, transmission through a potentially adversarial channel, and measurement. Bottom: Classical data post-processing encompasses data sifting, preprocessing (the focus of this work), information reconciliation, and privacy amplification. Parameter estimation, utilizing data statistics, is also crucial. For simplicity, procedures like authentication and error verification \cite{Fung2010Finite} are omitted.}
	\label{fig:QKDflowchart}
\end{figure}

While hardware advancements undoubtedly contribute to improving QKD performance \cite{RevModPhys.92.025002}, the development of refined post-processing schemes presents a more economically viable path. Standard post-processing methods typically encompass data sifting, information reconciliation, and privacy amplification. Initially, raw keys are obtained after data sifting, containing both bit and phase errors. The goals of information reconciliation and privacy amplification are to correct bit and phase errors separately, utilizing standard error correction method \cite{Shannon_mathematical_1948}. Shor and Preskill \cite{Shor_simple_2000} derive the key formula for the BB84 protocol:
\begin{equation}
	1 - h(\delta_b) - h(\delta_p)
\end{equation}
where $\delta_b$ and $\delta_p$ represent the bit and phase error rates, respectively, and $h$ is the entropy function. The information-theoretical security of these processes is grounded in entanglement distillation \cite{Lo1999Unconditional}, complementary approaches \cite{koashi2005simple}, or entropic approaches \cite{Devetak_2005, Renner_security_2008}. Among these, the security proof based on entanglement distillation stands out due to its broader applicability, extending beyond QKD to encompass entanglement distillation and quantum networks. Initially conceived for distilling perfect Einstein-Podolsky-Rosen (EPR) pairs between Alice and Bob, this proof was later ingeniously extended to the prepare-and-measure protocol \cite{Shor_simple_2000}. In the practical prepared-and-measure protocols, one-time pad (OTP) encryption is usually applied \cite{koashi2005simple, Huang_stream_2022}. Nonetheless, the role of OTP in classical communication warrants further investigation. 

While data sifting, information reconciliation, and privacy amplification are often standardized, the preprocessing step is where various post-processing schemes strive to enhance QKD performance, thus becoming a focal point for innovation. Preprocessing seeks effective methods to further select keys and alter the bit and phase errors, aiming to improve key rate and error thresholds. A particularly effective method, classical advantage distillation, leverages two-way communication to establish secret keys in scenarios where one-way communication proves insufficient \cite{Maurer_1993_secret}. Gottesman and Lo \cite{Gottesman_proof_2003} harnessed two-way communication to enhance the performance of QKD protocols. Their innovative introduction of the ``B step" as a preprocessing operation on raw keys, applied prior to standard information reconciliation and privacy amplification, paved the way for further advancements. Building upon this, Chau's introduction of adaptive privacy amplification \cite{Chau_2002_practical} further pushed the boundaries, achieving the current state-of-the-art results in tolerable error rates. 

In addition to two-way communication techniques, a novel protocol has emerged that involves Alice deliberately adding noise to her raw keys \cite{PhysRevLett.95.080501}. This strategy, aimed at reducing Eve's knowledge of the keys, is also considered a form of preprocessing. It can be implemented using one-way classical communication, illustrating that innovative one-way schemes can significantly enhance the performance of QKD. 

Although various methods have been proposed, achieving optimal post-processing remains an important open question. Given the considerable flexibility inherent in preprocessing techniques, the development of a comprehensive framework that both unifies and enhances these diverse methods becomes imperative. Such a framework would not only deepen our understanding of existing protocols but also serve as a springboard for refining them, ultimately boosting the performance of QKD protocols.

In this paper, we propose an advantage distillation framework for QKD, demonstrating its capacity to seamlessly integrate previous protocols, such as the B step and adding noise, while exhibiting remarkable generality. This integration lays the groundwork for a fundamental understanding that can be harnessed to elevate key rates and error thresholds across a spectrum of preprocessing schemes. To showcase the potency of our framework, we apply it in conjunction with classical linear codes. We rigorously derive key rate formulas for various scenarios, encompassing adding noise and schemes both with and without OTP encryption. Our findings underscore the enhanced efficiency of QKD protocols that forgo OTP encryption compared to those conventionally employing it.  Moreover, we apply our framework to $[n\ 1\ n]$ codes, deriving a tolerable error rate that aligns with the best-known results in the literature \cite{Chau_2002_practical}. Notably, while our framework seamlessly integrates previous protocols \cite{Watanabe2007keyrate}, we achieve an improvement in key rate through the systematic utilization of $[n\ 1\ n]$ codes and $[m\ m-1\ 2]$ codes. These insights collectively offer a methodical approach to refining current QKD protocols. 

The remainder of this paper is organized as follows. In Sec.~\ref{sec:prelim}, we review the basic concepts of error correction and the security proof based on entanglement distillation. In Sec.~\ref{sec:qkd_adv_distill_framework}, we introduce our QKD advantage distillation framework and apply it to different protocols, including those based on classical linear codes. In Sec.~\ref{sec:qkd_adv_distill_key_rate}, we analyze the role of error-correcting codes in advantage distillation and derive key rate formulas under different scenarios. Our analysis shows that omitting OTP encryption can yield a higher key rate.  We also derive the key rate formula for incorporating the adding noise method into our framework. In Sec.~\ref{sec:application}, we employ our framework to derive the error tolerance threshold that aligns with the best-known result and improves the key rate of previous works. In Sec.~\ref{sec:conclusion}, we provide the conclusion and outlook.

\section{Preliminaries}\label{sec:prelim}
In this section, we introduce the basic concepts and results of error correction and how they fit into proving the security of QKD protocols under the entanglement distillation picture. 

\subsection{Classical linear code}
We introduce classical linear code \cite{Nielsen_quantum_2002}. A linear code $C$ encoding $k$-bit message into an $n$-bit code space is specified by the $n \times k$ \textit{generator matrix} $G$ whose elements belong to \{0,1\}. The matrix $G$ maps messages to their encoded equivalent. Thus, the $k$-bit message $x$ is encoded as $Gx$. The multiplication and plus operation in linear code is modulo 2. The set of codewords generated by $G$ is also denoted as $C$, in which
\begin{equation}\label{eq:code_space}
	C = \{Gx : x\in \{0,1\}^k\}.
\end{equation} 

Denote $|x|$ as the Hamming weight of a codeword $x$, i.e., the number of nonzero entries of $x$. The Hamming distance between two codewords $x,y$ as $d(x,y) \equiv |x + y|$. The distance $d$ of a linear code $C$ can be defined as 
\begin{equation}
	d \equiv \min_{x,y \in C, x\neq y} d(x,y) = \min_{x\in C, x\neq 0} |x|.
\end{equation}
We denoted $C$ as $[n\ k\ d]$ code.

Another important matrix is the $(n-k) \times n$ \textit{parity check matrix} $H$ satisfying $HG = 0$. From Eq.~\eqref{eq:code_space}, we have
\begin{equation}
	Hx = 0, \forall x \in C.
\end{equation}
If $Hx \neq 0$, then some errors must cause the encoded message to no longer belong to $C$.

We introduce some definitions that are used substantially in this work.  The \textit{tag} of an $n$-bit string $x$ is defined as $t = Hx$. An \textit{error pattern} is a bit string $e$ of length $n$. The tag of an error pattern $e$ is defined as \textit{error syndrome} $s = He$, which has length $n-k$. There are $2^{n-k}$ different syndrome $\{s^0,s^1,\cdots,s^{2^{n-k}-1}\}$, and for each syndrome $s^j$, there are $2^k$ error pattern $\{e^j_0,e^j_1,\cdots,e^j_{2^k-1}\}$ satisfying $He^j_i = s^j$. We will use error syndromes to classify error patterns.

For a linear code $C$ with parity check matrix $H$ and generator matrix $G$, its \textit{dual code} $C^{\perp}$ is defined with generator matrix $H^T$ and parity check matrix $G^T$. A code and its dual code are useful in the security proof based on entanglement distillation, where the parity check matrix $H$ of code $C$ is used to classify bit error patterns and the parity check matrix $G^T$ of code $C^{\perp}$ is used to classify phase error patterns. For more discussions, see Section~\ref{sec:one-way-sec} and Section~\ref{sec:qkd_adv_distill_key_rate}.

Consider the $[7\ 4\ 3]$ code as an illustrative example, with the parity check matrix $H$ given by:
\begin{equation}
    H = \begin{pmatrix}
        1 & 0 & 0 & 1 & 1 & 0 & 1 \\
        0 & 1 & 0 & 1 & 0 & 1 & 1 \\
        0 & 0 & 1 & 0 & 1 & 1 & 1
    \end{pmatrix}.
\end{equation}

The quantum circuit for extracting the tag $Hx$ from a bit string $x$ is provided in Fig.~\ref{fig:linear_code}. Assuming the initial state of the data qubits is $\ket{x}$ in the computational basis, the tag can be extracted using CNOT gates from data qubits to ancillary qubits. 

\begin{figure}[!h]
\begin{quantikz}[row sep=0.2cm, column sep=0.2cm]
\lstick[7]{Data qubits in \\computational basis $\ket{x}$}& \ctrl{7} &&&&&&&&  \\
    && \ctrl{7} &&&&&&&  \\
    &&& \ctrl{7} &&&&&& \\
    &&&& \ctrl{5} &&&&& \\
    &&&&& \ctrl{5} &&&& \\
    &&&&&& \ctrl{4} &&&  \\
    &&&&&&& \ctrl{3} &&  \\
    \lstick{$\ket{0}$}&\targ{}&&& \targ{} &\targ{}&& \targ{} &\meter{} &\setwiretype{c} \rstick[3]{Tag $t=Hx$}\\
    \lstick{$\ket{0}$}&& \targ{} && \targ{} && \targ{} &\targ{}&\meter{}  &\setwiretype{c} \\
    \lstick{$\ket{0}$} &&& \targ{} && \targ{} & \targ{} & \targ{} &\meter{} &\setwiretype{c} \\
\end{quantikz}

\caption{Quantum circuit for extracting the tag $t = Hx$ from data qubits $\ket{x}$. The circuit employs the parity check matrix $H$ of the $[7\ 4\ 3]$ code. Ancillary qubits are initialized to $\ket{0}$. The operations on each ancillary qubit correspond to a row of $H$. Ancillary qubits are measured in the computational basis at the end, yielding the tag $t=Hx$ of the data $x$.}
\label{fig:linear_code}
\end{figure}
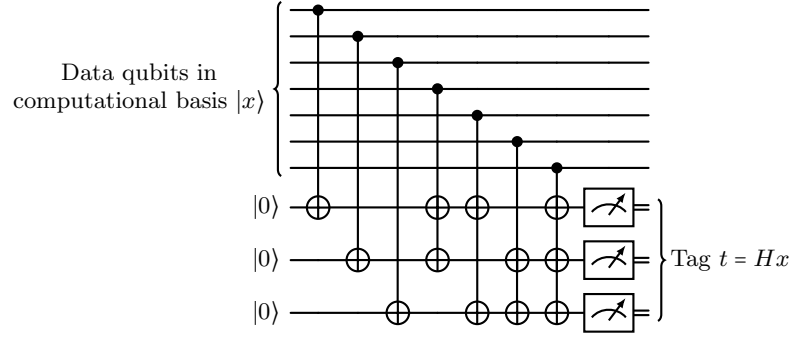

Notably, each CNOT gate has one of the data qubits as the control qubit, showcasing a feature of linear codes. This property does not hold for non-linear operations, such as the CCNOT operation depicted in Fig.~\ref{fig:CCNOT}, which might be utilized in some non-linear codes. While we do not delve into the application of non-linear codes in this work, it is worth mentioning that some of them can be integrated into our framework if they are classically replaceable operations \cite{Liu2022classically}. We will discuss the classically replaceable operation in Sec.~\ref{sec:one-way-sec}.

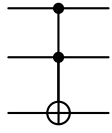
\begin{figure}[!h]
\begin{quantikz}
    & \ctrl{2} &  \\
    & \ctrl{1} &  \\
    & \targ{} & 
\end{quantikz}
\caption{The CCNOT operation, also known as the Toffoli gate, exhibits non-linearity in its treatment of the two control qubits.}
\label{fig:CCNOT}
\end{figure}

\subsection{$\varepsilon$-smallest probable set and error correction}
We discuss basic concepts about error correction. Suppose Alice and Bob hold two $n$-bit strings $x,y$, the error pattern $e$ is defined as $e = x + y$. We show how Alice and Bob can identify the error pattern $e$. The error pattern $e$ can be seen as a realization of a random variable $X$ distributed on $\{0,1\}^n$, according to a probability distribution $p_X$. To analyze the behavior of error patterns in the asymptotic limit, we first define the $\varepsilon$-smallest probable set.

\begin{definition}[$\varepsilon$-smallest probable set]
Given $0\leq \varepsilon \leq \frac{1}{2}$, a random variable $X\in \mathcal{X}$ which has a finite range $|\mathcal{X}|<\infty$ and probability distribution $p_X$, the $\varepsilon$-smallest probable set $\mathcal{T}^{\varepsilon}_X$ of $X$ is defined as the smallest set  of values such that a realization of $X$ lies in it with a failure probability no larger than $\varepsilon$,
\begin{equation}
\begin{split}
\mathcal{T}^{\varepsilon}_X = \argmin_{\mathcal{S}}|\mathcal{S}|,\\
s.t.\Pr(X\in \mathcal{S})\geq 1-\varepsilon.
\end{split}
\end{equation}
\end{definition}

If $X = (X_1,X_2,\cdots,X_n)$ and $X_1,X_2,\cdots,X_n$ are $i.i.d.$, the cardinality of $\varepsilon$-smallest probable set can be upper bounded by 
\begin{equation}\label{eq:typical_set}
	|\mathcal{T}^{\varepsilon}_X| \le 2^{n (h(X_1) + \varepsilon)}
\end{equation} 
for $n$ sufficiently large \cite{cover1999elements}. Here, the entropy of random variable $h(X_1)$ is defined as
\begin{equation}
	h(X_1) = - \sum_{x}  p(X_1=x) \log p(X_1=x).
\end{equation} 
We also define
\begin{equation}
	h(x) = - x \log x - (1-x) \log (1-x)
\end{equation}
to denote the entropy function of a real number $x$, where $0 \le x \le 1$.

To succeed in correcting the error with a failure probability $\varepsilon$, we only need to guarantee success when the error pattern is in the $\varepsilon$-smallest probable set $\mathcal{T}^{\varepsilon}_X$. To identify the error pattern, we introduce universal hash functions.

\begin{definition}[Universal hash family \cite{CARTER1979143, WEGMAN1981265}]
	A family of functions $\cF$ mapping elements $e$ in a space $\cT$ to another space $\cS$ is called a universal hash family if the probability of a randomly chosen hash function $f\in \cF$ outputting the same hashing result for any two different strings is upper bounded by
	\begin{equation}
		\forall e_i \neq e_j \in \cT, \Pr_{f\in\cF}[f(e_i) = f(e_j)] \le \frac{1}{|\cS|}.
	\end{equation}
\end{definition}

In our scenario, the elements of $\mathcal{T}$ comprise $n$-bit error patterns and $|\mathcal{S}| = 2^k$, with $k$ denoting the length of the tag. A crucial step in error pattern identification is selecting universal hashing functions, and random $k \times n$ hashing matrices can be employed for this purpose \cite{Huang_stream_2022}. 

To identify the error pattern $e$, they jointly and randomly select a hash function $f \in \mathcal{F}$. Then, they compute the tag $s_x = f(x)$ and $s_y = f(y)$ independently. Alice then transmits $s_x$ to Bob through an authenticated channel, after which Bob calculates the error syndrome $s = f(e) = f(x \oplus y) = s_x \oplus s_y$. If $e$ is the unique error pattern in $\mathcal{T}$ satisfying $f(e) = s$, then $e$ can be successfully identified. The probability of failure is bounded by 
\begin{equation}\label{eq:fail_prob}
	\sum_{e' \in \cT \backslash \{e\}} \Pr[f(e') = f(e)] \le \frac{|\mathcal{T}|}{|\mathcal{S}|}.
\end{equation} 
Hence, $e$ can be identified with a probability of $1-\frac{|\mathcal{T}|}{|\mathcal{S}|}$.

From Eq.~\eqref{eq:fail_prob}, to guarantee small failure probability $\delta$ on successfully identifying the error pattern in set $\cT$, the length of the tag should satisfy
\begin{equation}\label{eq:tag_length}
	k \ge \log |\cT| + \log \frac{1}{\delta}.
\end{equation} Combined with Eq.~\eqref{eq:typical_set}, we derive the following lemma.
\begin{lemma}[One-way error correction \cite{Shannon_mathematical_1948, cover1999elements}]\label{lem:error_space}
Suppose Alice holds $n$-bit string $x \in C$, Bob holds n-bit string $x + e$, $e$ is a realization of a random variable $X$, and $k$ is the length of the tag. To identify the error pattern $e$ with probability $1-\ve$, $\frac{k}{n}$ can be arbitrarily close to $\frac{\log |\cT^{\varepsilon}_n|}{n}$ for sufficiently large $n$. If $X=(X_1,X_2,\cdots,X_n)$ and $X_i$s are $i.i.d$, then for any $\ve > 0$, $\frac{k}{n}$ can be arbitrarily close to $h(X_1)$ for sufficiently large $n$.
\end{lemma}

\begin{proof}
	From Eq.~\eqref{eq:tag_length}, to identify $e$ with failure probability $\ve + \delta$, the length of the tag can be chosen as $\frac{k}{n} \ge \frac{\log |\cT^{\varepsilon}_n| + \log \frac{1}{\delta}}{n}$. $\delta$ can be arbitrarily close to 0 and $\frac{k}{n}$ can be arbitrarily close to $\frac{\log |\cT^{\varepsilon}_n|}{n}$, provided $n$ is sufficiently large.
		
	If $X = (X_1,X_2,\cdots,X_n)$, $X_i$s are $i.i.d$, then by Eq.~\eqref{eq:typical_set}, $\frac{k}{n} \le h(X_1) + \varepsilon$ for $n$ sufficiently large. Note that $\varepsilon$ can be chosen arbitrarily small. Thus, we draw the conclusion. 
\end{proof}

In this error correction scheme, Alice transmits information to Bob, enabling him to identify and correct the error pattern $e$, ultimately resulting in both parties possessing identical key strings. We designate such a scheme as a one-way error correction scheme.

\subsection{Security proof of QKD based on entanglement distillation}\label{sec:one-way-sec}
We present the security proof of QKD based on entanglement distillation. First, we describe an entanglement distillation protocol that utilizes one-way classical communication \cite{Bennett1996mixedstate, Lo1999Unconditional}. Then, we demonstrate that this protocol, when followed by computational basis measurements, can be reduced to the prepare-and-measure QKD protocol, thereby establishing its security \cite{Shor_simple_2000, Huang_stream_2022}.

\subsubsection{One-way entanglement distillation}
 Denote the perfect EPR pair as $\ket{\psi_{00}} = \frac{1}{\sqrt{2}} (\ket{00} + \ket{11})$. Entanglement distillation can be formalized as this:  After distributing several noisy EPR pairs through a noisy channel to two parties, Alice and Bob measure some of the EPR pairs to estimate the error, with a $2N$-qubits state $\rho^N$ remained. Based on the error estimation, they perform quantum error correction to transform $\rho^N$ into a state $\sigma^M$ which can be arbitrarily close to $\ket{\psi_{00}}^{\otimes M}$ in the large $N$ limit. Then, they measure the state $\sigma^M$ on the $Z$ basis and obtain secure keys. 

To perform the error estimation, Alice and Bob can perform $X$ and $Z$ measurements on subsets of EPR pairs and employ the random sampling method \cite{Fung2010Finite} to obtain the bit error rate $\delta_b$ and phase error rate $\delta_p$ of $2N-$qubits state $\rho^N$, which is defined as:
\begin{equation}\label{eq:bp_errorN}
\begin{split}
	\delta_b = \frac{1}{2}(1 - \frac{1}{N}\tr [\rho^N \sum_{i=1}^N(\sigma^i_z \otimes \sigma^i_z)]), \\
	\delta_p = \frac{1}{2}(1 - \frac{1}{N}\tr [\rho^N \sum_{i=1}^N(\sigma^i_x \otimes \sigma^i_x)]),
\end{split}	
\end{equation}
in which $\sigma_z^i$($\sigma_x^i$) denotes the Pauli operator that perform $\sigma_z$($\sigma_x$) on site $i$ and acts trivially on other sites. Error estimation, which is equivalent to the parameter estimation step in BB84 protocols, can be treated as a separated process preceding the subsequent distillation protocols.

Denote the Bell basis as $\ket{\psi_{ab}} = I \otimes X^a Z^b \ket{\psi_{00}}$, $a,b \in \{0,1\}$. That is,
\begin{equation}
	\begin{split}
		\ket{\psi_{00}} = \frac{1}{2}( \ket{00} + \ket{11}), \\
		\ket{\psi_{10}} = \frac{1}{2}( \ket{01} + \ket{10}), \\
		\ket{\psi_{01}} = \frac{1}{2}( \ket{00} - \ket{11}), \\
		\ket{\psi_{11}} = \frac{1}{2}( \ket{01} - \ket{10}). 
	\end{split}
\end{equation}
We express the complete state of $N$ noisy EPR pairs in the Bell basis. In Appendix \ref{app:non-Bell-diagonal}, we show that the off-diagonal terms do not impact the distillation process. Thus, we can restrict our attention to the Bell-diagonal terms. Additionally, the bit and phase errors in each qubit can be made independent and identically distributed ($i.i.d.$) through random permutation \cite{Renner209definetti}. Therefore, without loss of generality, we consider the input state to be a Bell-diagonal state of the form $\rho^{\otimes N}$, where
\begin{equation}
	\rho = \diag(p_{00},p_{01},p_{10},p_{11})
\end{equation}
in the Bell basis $\ket{\psi_{00}},\ket{\psi_{01}},\ket{\psi_{10}},\ket{\psi_{11}}$. From Eq.~\eqref{eq:bp_errorN}, we have
\begin{equation}
\begin{split}
	\delta_b = p_{10} + p_{11},\\
	\delta_p = p_{01} + p_{11}.
\end{split}
\end{equation}

As $\rho$ is Bell-diagonal, $\rho^{\otimes N}$ is in the state $\bigotimes_{i=1}^N \psi_{a_i,b_i}$ with a probability of $\prod_{i=1}^N p_{a_i,b_i}$. The bit error pattern is defined as $e_b = a_1a_2\cdots a_N$, and the phase error pattern is defined as $e_p = b_1b_2\cdots b_N$. Alice and Bob apply CNOT gates from $\rho^{\otimes N}$ to the ancillary qubits according to a hashing matrix $H$ and measure the ancillary qubits to identify the bit error pattern. Subsequently, Alice sends her measurement results to Bob. Bob obtains the bit error syndrome $s_b=Me_b$ and deduces the bit error with high probability. According to Lemma \ref{lem:error_space}, the length of the tag $\frac{k}{n}$ is $h(\delta_b)$ asymptotically.

In entanglement distillation, bit error correction and phase error correction need to be implemented sequentially. Thus, we need to ensure that the two error correction procedures do not affect each other. As stated in the following lemma, this is guaranteed when we use perfect EPR pairs as ancillary qubits.
\begin{lemma}[Bit- and phase-error-correction decoupling \cite{Lo_2003}]\label{lem:bp_decouple}
	Using EPR pairs as ancilla qubits, bit error correction does not affect phase errors and vice versa.
\end{lemma}

\begin{proof}
Consider the effect of bit error correction on the phase error as an example. As shown in Fig.~\ref{fig:proof_decouple}, if Alice and Bob use EPR pairs as ancillary qubits, the equivalent circuit indicates that, in the phase error space, Alice and Bob \textit{simultaneously} apply the XOR operation or do not apply it at all, which keeps the phase error unchanged.

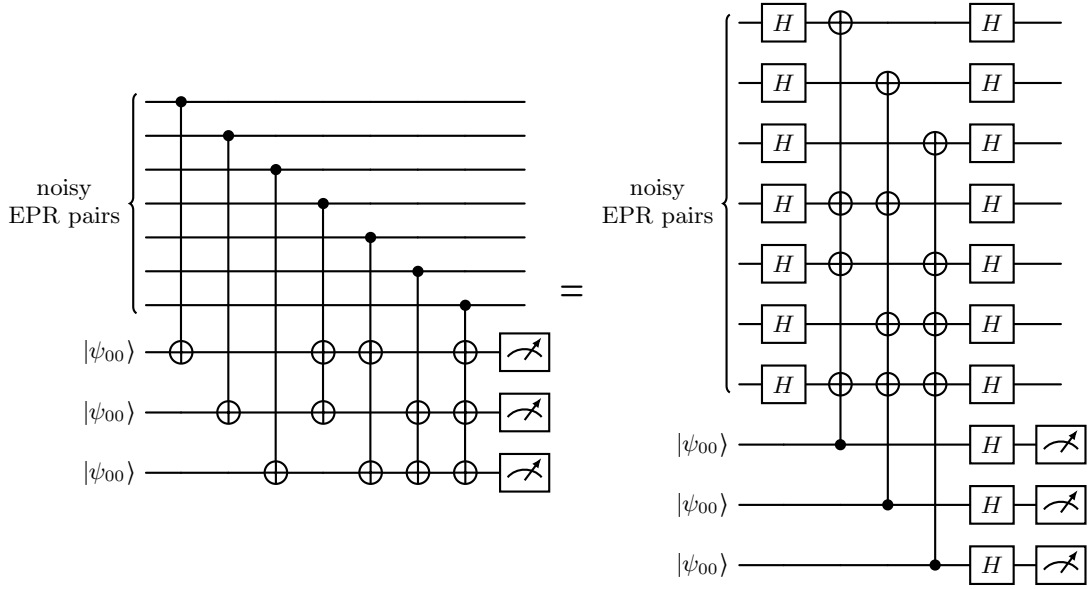
\begin{figure}[!h]
	\begin{quantikz}[row sep=0.3cm, column sep=0.3cm]
		\lstick[7]{noisy \\EPR pairs}& \ctrl{7} &&&&&&&  \\
		&& \ctrl{7} &&&&&&  \\
		&&& \ctrl{7} &&&&& \\
		&&&& \ctrl{5} &&&& \\
		&&&&& \ctrl{5} &&& \\
		&&&&&& \ctrl{4} &&  \\
		&&&&&&& \ctrl{3} &  \\
		\lstick{$\ket{\psi_{00}}$}&\targ{}&&& \targ{} &\targ{}&& \targ{} &\meter{}  \\
		\lstick{$\ket{\psi_{00}}$}&& \targ{} && \targ{} && \targ{} &\targ{}&\meter{}  \\
		\lstick{$\ket{\psi_{00}}$} &&& \targ{} && \targ{} & \targ{} & \targ{} &\meter{}  \\
	\end{quantikz} \raisebox{0.5\height}{\scalebox{2}{$=$}} 
	\begin{quantikz}[row sep=0.3cm, column sep=0.3cm]
		\lstick[7]{noisy \\EPR pairs}&\gate{H} & \targ{} &&& \gate{H} &  \\
		&\gate{H} && \targ{} && \gate{H} &  \\
		&\gate{H} &&& \targ{} & \gate{H} & \\
		&\gate{H} &\targ{} & \targ{} && \gate{H} & \\
		&\gate{H} &\targ{} && \targ{} & \gate{H} &  \\
		&\gate{H} & & \targ{} & \targ{} & \gate{H} &  \\
		&\gate{H} & \targ{} &\targ{}& \targ{} & \gate{H} &  \\
		\lstick{$\ket{\psi_{00}}$}&&\ctrl{-7}&&& \gate{H} &\meter{}  \\
		\lstick{$\ket{\psi_{00}}$}&&& \ctrl{-7} && \gate{H} &\meter{} \\
		\lstick{$\ket{\psi_{00}}$} &&&& \ctrl{-7} & \gate{H} &\meter{}  \\
	\end{quantikz}
	\caption{The quantum circuit for hashing and measuring the tag, along with its equivalent circuit. Only one side of Alice and Bob is illustrated. Here, the Hadamard gates transform phase errors into bit errors. Since the ancillary qubits are perfect EPR pairs, Alice and Bob simultaneously apply the XOR operation or not, ensuring that the phase error pattern remains unchanged.}
	\label{fig:proof_decouple}
\end{figure}

We can also directly prove this statement by considering the effect of the bilateral CNOT operation on EPR pairs:
\begin{equation}\label{eq:bilateralCNOT}
	(\CNOT_{13} \CNOT_{24}) \ketbra{\psi_{a,b}}_{12} \ketbra{\psi_{c,d}}_{34} (\CNOT_{13} \CNOT_{24})^{\dagger}= \ketbra{\psi_{a,b\oplus d}}_{12}\ketbra{\psi_{a\oplus c, d}}_{34}.
\end{equation}
Here, $\CNOT_{13}$ ($\CNOT_{24}$) corresponds to the CNOT operation on Alice's (Bob's) side. When the ancillary qubits are initially perfect EPR pairs, meaning $c = 0$ and $d = 0$, the subsequent CNOT operations only alter $c$, while keeping $d = 0$ unchanged. Thus, $b$ remains unchanged after each CNOT operation, ensuring that the phase error is unaffected.

After identifying the bit error pattern, Bob can apply $X$ gates to correct bit errors in some of the noisy EPR pairs. The $X$ operation does not affect the phase error since $X\ket{\psi_{ab}} = \ket{\psi_{a\oplus 1, b}}$. Thus, bit error correction does not influence the phase error pattern, and vice versa, due to the symmetry between bit and phase error correction.
\end{proof}

In the BB84 protocol, Alice and Bob are aware of both the bit and phase error rates. In the worst-case scenario, these errors are independent. Consequently, correcting the bit errors consumes $n h(\delta_b)$ EPR pairs, while correcting the phase errors consumes $n h(\delta_p)$ EPR pairs. Therefore, a total of $n h(\delta_b) + n h(\delta_p)$ EPR pairs are used to obtain $n$ perfect EPR pairs with high probability. The distillable rate, defined as the net gain of perfect EPR pairs divided by the number of noisy EPR pairs consumed, is given by:
\begin{equation}\label{eq:shor_preskill}
	r = 1 - h(\delta_b) - h(\delta_p).
\end{equation} 

If Alice and Bob acquire more detailed information about the error terms beyond just the bit error rate $\delta_b$ and the phase error rate $\delta_p$, the distillable rate can be further improved. For instance, by performing measurements in the $X$, $Y$, and $Z$ bases during parameter estimation—as in the six-state protocol \cite{Lo_proof_2001}—they can fully characterize all the diagonal terms of the noisy EPR pairs. This additional information increases the tolerable bit error rate up to $12.7\%$. Notably, this entanglement distillation protocol requires only one-way communication, as it leverages one-way error correction, as described in Lemma \ref{lem:error_space}.

\subsubsection{Reduction to prepare-and-measure QKD protocol}
Suppose Alice and Bob successfully distill $M$ perfect EPR pairs, $\ket{\psi_{00}}^{\otimes M}$. Then, Eve is only left with a state $\rho_E$ such that the total system is in a product state:
\begin{equation}
	\ketbra{\psi_{00}}^{\otimes M} \otimes \rho_E.
\end{equation}
After Alice and Bob measure the perfect EPR pairs in the computational basis, the state becomes:
\begin{equation}
 	2^{-M}\sum_{k \in \{0,1\}^M} \ketbra{kk}_{AB} \otimes \rho_E.
\end{equation}
Here, Eve's state $\rho_E$ is independent of the secret keys $k$, ensuring that the keys are secure. In practice, when Alice and Bob hold states that are very close to perfect EPR pairs, security can still be proven \cite{BenOr2005universalcomposable, Renner2005universallycomposable, Fung2010Finite}. Consequently, Eq.~\eqref{eq:shor_preskill} provides the key rate for QKD. The maximum tolerable bit error rate for this one-way protocol, which represents the threshold at which Alice and Bob can still obtain a net gain of secure keys, is $11\%$.

We call the process of bit error correction information reconciliation because it ensures that Alice and Bob ultimately share the same keys. To guarantee the keys are secure, we only need to show that the phase error pattern can be identified without actually calculating it \cite{Shor_simple_2000}. Consequently, the phase error correction can be replaced with privacy amplification, which introduces additional shared randomness to erase Eve's knowledge about the keys \cite{Huang_stream_2022}. The quantum key distribution process can then be succinctly summarized as information reconciliation followed by privacy amplification, as illustrated in Fig.~\ref{fig:shor_presill_qkd}.

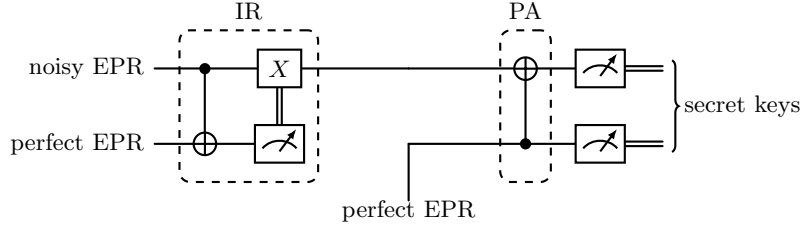
\begin{figure}[!h]
\begin{quantikz}
\lstick[1]{noisy EPR} & \ctrl{1} \gategroup[2,steps=2,style={dashed,rounded
corners, inner xsep=2pt}]{IR} & \gate{X} & & \targ{} \gategroup[2,steps=1,style={dashed,rounded
corners, inner xsep=2pt}]{PA}&\meter{} &\setwiretype{c}\rstick[2]{secret keys} \\
\lstick[1]{perfect EPR}  & \targ{} &\meter{} \wire[u][1]{c}& \setwiretype{n} & \ctrl{-1} \setwiretype{q}  & \meter{} &\setwiretype{c}\\
\setwiretype{n} &&&{\text{perfect EPR}} \wire[u][1]{q} &&&
\end{quantikz}
\caption{The QKD protocol consists of two essential steps: information reconciliation (IR) and privacy amplification(PA)  \cite{Shor_simple_2000, Huang_stream_2022}. Here, we focus on Bob's operations. In the information reconciliation step, Alice and Bob conduct hashing from noisy EPR pairs to perfect EPR pairs to extract the bit error syndrome(illustrated by the CNOT operation from noisy EPR to perfect EPR). Subsequently, Bob corrects his noisy EPR pairs based on the obtained syndrome, employing the controlled-X operation. They proceed to privacy amplification, which uses additional perfect EPR pairs to add randomness to their keys, ensuring the security of the generated keys (illustrated by the CNOT operation from perfect EPR to noisy EPR). Finally, the process concludes with measuring the noisy EPR pairs in the computational basis, obtaining the secret keys.}
\label{fig:shor_presill_qkd}
\end{figure}

The procedure in Fig.~\ref{fig:shor_presill_qkd} requires quantum computers. Fortunately, it is equivalent to the prepare-and-measure protocols \cite{Shor_simple_2000, Huang_stream_2022}. This reduction to classical operations can be understood through the framework proposed by \cite{Liu2022classically}, which defines the classically replaceable operation.
\begin{definition}[Classically Replaceable Operation, CRO \cite{Liu2022classically}]
	Consider operations right before a measurement. A CRO can be realized by first measuring the input state and then processing the outcomes classically.
\end{definition}

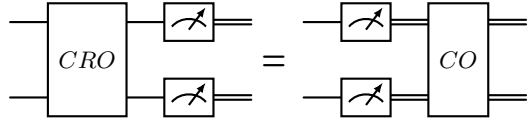
\begin{figure}[!h]
\begin{quantikz}
 & \gate[2]{CRO} & \meter{} & \setwiretype{c}  \\
 & & \meter{} & \setwiretype{c} \\
\end{quantikz}  \raisebox{0.7\height}{\scalebox{2}{$=$}} 
\begin{quantikz}
 &\meter{}&  \gate[2]{CO} \setwiretype{c}&   \\
 &\meter{}& \setwiretype{c}&\\
\end{quantikz} 
\caption{Illustration of CRO. A CRO operation followed by measurement in the computational basis is equivalent to measurement in the computational basis followed by classical operation (CO).}
\label{fig:CRO}
\end{figure}

The process of replacing quantum operations with classical operations is depicted in Fig.~\ref{fig:CRO}. In the QKD protocol, information reconciliation and privacy amplification are classified as CROs \cite{Huang_stream_2022, Liu2022classically}. Consequently, the final measurement in the QKD protocol can be moved to the beginning to obtain raw keys $x$ and $y$ on each side, followed by classical information reconciliation and privacy amplification, as illustrated in Fig.~\ref{fig:QKD_CRO}. Using perfect EPR pairs is thus reduced to employing pre-shared secret keys for encrypting communication via OTP. For information reconcillation, Alice uses a pre-shared secret key to encrypt the tag $s_x$ of her raw keys $x$, $s_x = Hx$. Bob can then use the same pre-shared key to decrypt Alice's tag, obtain the symdrome of bit error pattern $s_b = s_x \oplus s_y$, and correct the bit errors $e_b = x \oplus y$. For the adversary, Eve, the encrypted message appears as a random string, offering no information. For privacy amplification, Alice and Bob use pre-shared secret keys to enhance the secrecy of keys \cite{Huang_stream_2022}.  Ultimately, the QKD protocol is reduced to a prepare-and-measure protocol.

\begin{figure}[!h]
\raisebox{0.7\height}{(a)} 
\begin{quantikz}
\lstick[1]{noisy EPR} & \gate[2]{CRO}\gategroup[2,steps=2,style={dashed,rounded
corners, inner xsep=2pt}]{IR} & \gate[1]{CRO}& & \gate[2]{CRO}\gategroup[2,steps=1,style={dashed,rounded
corners, inner xsep=2pt}]{PA} &\meter{} &\setwiretype{c} \\
\lstick[1]{perfect EPR}  & \targ{} &\meter{} \wire[u][1]{c}& \setwiretype{n} & \setwiretype{q}  & \meter{} &\setwiretype{c}\\
\setwiretype{n} &&&{\text{perfect EPR}} \wire[u][1]{q} &&&
\end{quantikz} 

\raisebox{0.7\height}{(b)} 
\begin{quantikz}
\lstick[1]{raw key} \setwiretype{c} & \gate[2]{CO}\gategroup[2,steps=2,style={dashed,rounded corners, inner xsep=2pt}]{IR} & \gate[1]{CO}& & \gate[2]{CO}\gategroup[2,steps=1,style={dashed,rounded corners, inner xsep=2pt}]{PA} & \\
\lstick[1]{secret key} \setwiretype{c} & & \ctrl[vertical wire=c]{-1}& \setwiretype{n} & \setwiretype{c}  & \\
\setwiretype{n} &&&{\text{secret key}} \wire[u][1]{c} &&&
\end{quantikz} 

\caption{Reducing the QKD protocol to the prepare-and-measure protocol. (a) Information reconciliation (IR) and privacy amplification (PA) are classified as CROs; thus, the measurement can be moved forward to the beginning. (b) After moving the measurement forward, all operations become classical operations (CO). The perfect EPR becomes the secret key shared by Alice and Bob. Finally, the QKD protocol is reduced to a prepare-and-measure protocol.}
\label{fig:QKD_CRO}
\end{figure}
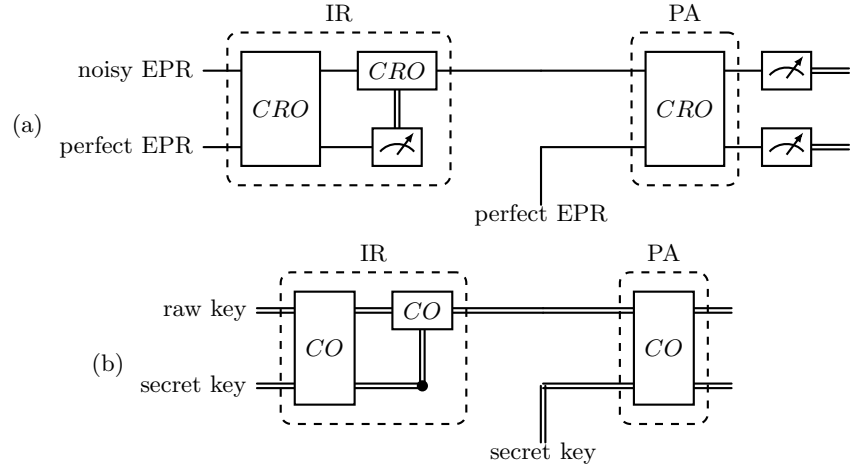

\section{Advantage distillation for quantum key distribution} 
As discussed above, Alice and Bob distill perfect EPR pairs through one-way communication and perform computational basis measurements to obtain their secret key. This process ultimately reduces to classical information reconciliation and privacy amplification using standard one-way hashing techniques. However, not all key distillation schemes rely solely on these two steps. As illustrated in Fig.~\ref{fig:QKDflowchart}, additional preprocessing steps that allow both one-way and two-way communication are essential for various schemes to enhance both the tolerable error rate and the key rate.

An illustrative approach is the B step \cite{Gottesman_proof_2003, Ma_decoy_2006}, which involves segregating all raw keys into distinct groups. Subsequently, the bit error syndrome is determined in each group through linear hashing and two-way classical communication. Groups exhibiting high bit error rates are discarded, enhancing the overall key rate and error tolerance. This operation can be reduced to a prepare-and-measure protocol \cite{Gottesman_proof_2003} and can be seen as a preprocessing step for refining the raw keys before information reconciliation and privacy amplification.

Adding noise is another example where Alice deliberately introduces noise into her raw keys \cite{PhysRevLett.98.020502, PhysRevLett.95.080501}. This intentional noise insertion diminishes Eve's knowledge of the keys, consequently increasing the key rate. Additionally, this method only requires one-way classical communication. While adding noise has not been proven secure through the entanglement distillation method, it has been shown that perfect EPR pairs are not a prerequisite for secure keys. Instead, private states are sufficient to establish secure keys and can be used to demonstrate the security of adding noise \cite{PhysRevLett.94.160502}.

\subsection{Advantage distillation framework} \label{sec:qkd_adv_distill_framework}
Numerous preprocessing techniques can be applied before Alice and Bob perform information reconciliation and privacy amplification. Therefore, developing a general framework to incorporate these techniques is necessary and helpful for further protocol design. Here, we propose the advantage distillation framework, which generalizes the methods above, following the process illustrated in Fig.~\ref{fig:QKDflowchart}. The protocols in our framework, developed within the entanglement distillation scenario, are equivalent to prepare-and-measure QKD protocols. This equivalence arises because all operations allowed in our framework are CROs, allowing a reduction similar to that in Section \ref{sec:one-way-sec}. Consequently, the distillable rate of EPR pairs is equivalent to the key rate of the corresponding QKD protocols.

As depicted in Fig.~\ref{fig:ad_framework}, our advantage distillation framework comprises three steps. Alice and Bob initially have noisy EPR pairs $\rho^{\otimes N}$ \cite{Bennett1996mixedstate}. First, they proceed to a preprocessing operation, which must be a CRO. Next, they discard some EPR pairs. Then, they perform information reconciliation to correct the bit errors. Finally, they proceed to privacy amplification. Ancillary qubits are utilized and measured at each step, and the measurement results can control subsequent operations classically.

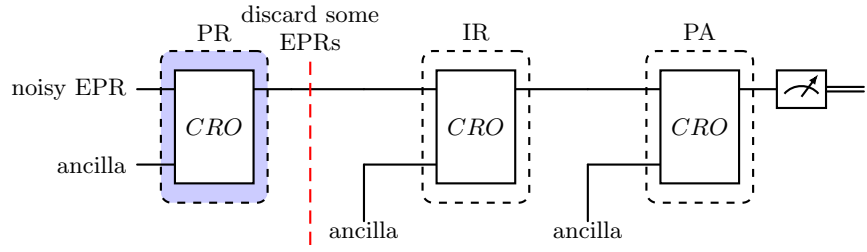
\begin{figure}[!h]
\begin{quantikz}
\lstick[1]{noisy EPR} &\gate[2]{CRO}\gategroup[2,steps=1,style={dashed,rounded corners,fill=blue!20, inner xsep=2pt},background,label style={label position=above,anchor=north,yshift=0.25cm}]{PR}  &\slice{discard some \\EPRs}&& \gate[2]{CRO}\gategroup[2,steps=1,style={dashed,rounded
corners, inner xsep=2pt}]{IR}   & & \gate[2]{CRO}\gategroup[2,steps=1,style={dashed,rounded
corners, inner xsep=2pt}]{PA}  &\meter{} &\setwiretype{c} \\
\lstick[1]{ancilla}  &&\setwiretype{n} &&  \setwiretype{q} & \setwiretype{n} & \setwiretype{q} & \setwiretype{n}  &\\
\setwiretype{n} &&&{\text{ancilla}}\wire[u][1]{q}&\setwiretype{n} &{\text{ancilla}} \wire[u][1]{q} &\setwiretype{n}&&
\end{quantikz} 
\caption{The advantage distillation framework. Initially, Alice and Bob preprocess (PR) the noisy EPR pairs. Subsequently, they discard certain EPR pairs based on the measurement results from preprocessing. They then proceed to information reconciliation (IR) and privacy amplification (PA). Finally, they perform computational basis measurements to obtain the secure key. The type of ancillary qubits is not specified and may vary in different protocols. Ancillary qubits are measured in each step, and all operations can be adaptive based on previous measurement results. We omit the measurement of ancillary qubits and the classical control wires in the figure. The operations must be CRO to enable the reduction to a prepare-and-measure protocol.}
\label{fig:ad_framework}
\end{figure}

We make some remarks on our framework. Firstly, it is possible to conduct multiple preprocessing rounds. As long as a CRO is performed in each round, the preprocessing eventually becomes a CRO. The CRO operation enables moving the measurement forward, effectively reducing the scheme to a prepare-and-measure protocol. Secondly, all the operations can be adaptive based on the measurement outcomes from previous rounds. Thirdly, some EPR pairs may be discarded before entering the information reconciliation process, similar to what is done in the B step. Lastly, the ancillary qubits are not specified, leaving room for devising different protocols. As we will show later, different types of ancillary qubits lead to different key rates, and there are better choices than perfect EPR pairs.

We exemplify our protocol using two previously mentioned examples to illustrate the generality of our approach. Adding noise \cite{PhysRevLett.95.080501, PhysRevLett.98.020502} is achieved by utilizing ancillary qubits in the form of $\sqrt{1-p}\ket{0} + \sqrt{p}\ket{1}$ and applying a CNOT gate from the ancillary qubits to the noisy EPR pairs, as demonstrated in Fig.~\ref{fig:adding_noise}. After measuring the ancillary qubits in the computational basis, the process is equivalent to adding a bit flip to the noisy EPR pair with a probability of $p$. 

\begin{figure}[!h]
	\begin{quantikz}
\lstick[1]{noisy EPR} & \targ{}\gategroup[2,steps=2,style={dashed,rounded corners,fill=blue!20, inner xsep=2pt},background,label style={label position=above,anchor=north,yshift=0.25cm}]{Adding noise} &  &  \\
\lstick[1]{ancilla \\ $\sqrt{1-p}\ket{0} + \sqrt{p}\ket{1}$}  & \ctrl{-1} & \meter{} & \setwiretype{c} 
\end{quantikz}
\caption{Incorporating adding noise \cite{PhysRevLett.98.020502, PhysRevLett.95.080501} into advantage distillation framework. The ancillary qubit is in the form of $\sqrt{1-p}\ket{0} + \sqrt{p}\ket{1}$ and a CNOT gate is applied from the ancillary qubit to the noisy EPR pair. After measuring the ancillary qubit in the computational basis, this process is equivalent to introducing a bit flip to the noisy EPR pair with a probability of $p$.}
\label{fig:adding_noise}
\end{figure}
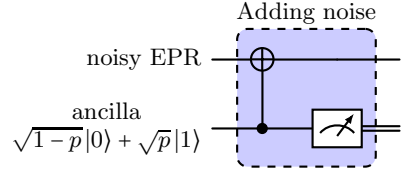

The B step \cite{Gottesman_proof_2003, Ma_decoy_2006} can also be incorporated into our framework, as depicted in Fig.~\ref{fig:B_step}. We present a single round B step for illustration. During this round, a CNOT operation is performed on a pair of noisy EPR pairs, followed by a measurement on one of the two EPR pairs to extract parity information. If the parity differs between Alice and Bob's sides, the remaining EPR pairs are discarded. Conversely, the remaining EPR pair is retained if the parity is the same. As remarked before, conducting the B step in multiple rounds is also a CRO and can be incorporated into our framework. 

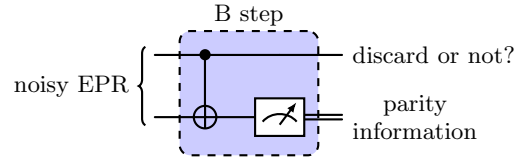
\begin{figure}[!h]
	\begin{quantikz}
\lstick[2]{noisy EPR} & \ctrl{1}\gategroup[2,steps=2,style={dashed,rounded corners,fill=blue!20, inner xsep=2pt},background,label style={label position=above,anchor=north,yshift=0.25cm}]{B step} &  &  \rstick[1]{discard or not?}\\
& \targ{} & \meter{} & \setwiretype{c} \rstick[1]{parity \\information}
\end{quantikz}
\caption{Incorporating the B step \cite{Gottesman_proof_2003, Ma_decoy_2006} into the advantage distillation framework. A CNOT operation is conducted on a pair of noisy EPR pairs. Subsequently, a measurement is performed on one of the two EPR pairs to extract parity information. If the parity differs between the two sides of Alice and Bob, the remaining EPR pair is discarded. Conversely, if the parity is the same, the EPR pair is retained.}
\label{fig:B_step}
\end{figure}

\subsection{Advantage distillation based on classical linear code}
Now, we develop the advantage distillation framework based on classical linear code. Notice that the B step corresponds to the $[2\ 1\ 2]$ code, which is not always optimal compared to other codes \cite{nonlinear_footnote}.  Thus, we incorporate the fundamental concept of the B step and extend it to encompass a general error-correcting code, as illustrated in Fig.~\ref{fig:ad_qec_frame} and Fig.~\ref{fig:hashing_detail} \cite{nonlinear_footnote}. 

The EPR pairs are initially divided into groups, each containing $n$ EPR pairs. Let $e_b$ and $e_p$ denote the bit error pattern and phase error pattern within a group. The linear hashing matrix $H$ is applied in each group. Subsequently, Alice and Bob engage in two-way classical communication to derive the bit error syndrome $s_b = He_b$. Based on the error syndrome $s_b$ and the corresponding key rate associated with this syndrome, as derived in Theorem \ref{thm:kr_OTP} and Theorem \ref{thm:kr_noOTP}, they decide whether to discard the group of EPR pairs. Following this decision, information reconciliation and privacy amplification are applied, and the secure keys are obtained.

\begin{figure}[!h]
\begin{quantikz}
\lstick[3]{noisy EPR} &\qwbundle{n} & \gate[4]{CRO}\gategroup[4,steps=1,style={dashed,rounded corners,fill=blue!20, inner xsep=2pt},background,label style={label position=above,anchor=north,yshift=0.25cm}]{Hashing}  &\slice{discard some \\ EPRs}& & \gate[2]{CRO}\gategroup[2,steps=1,style={dashed,rounded
corners, inner xsep=2pt}]{IR}   & & \gate[2]{CRO}\gategroup[2,steps=1,style={dashed,rounded
corners, inner xsep=2pt}]{PA}  &\meter{} &\setwiretype{c} \\
  &\setwiretype{n} \vdots & &\setwiretype{n}& \wire[d][2]{q}  &  \setwiretype{q} &\setwiretype{n} \wire[d][2]{q}&\setwiretype{q} &\setwiretype{n} & \\
  &\qwbundle{n} & &\setwiretype{n}&&   & & & & \\
\lstick[1]{ancilla}  && &\setwiretype{n}& {\text{ancilla}} &  & {\text{ancilla}} &  &
\end{quantikz} 
\caption{The advantage distillation framework based on classical linear code. Firstly, Alice and Bob divided the noisy EPR pairs into groups, each containing $n$ EPR pairs. They then apply the parity matrix of the classical linear code to each group, extracting the parity information. Then, they engage in two-way classical communication to derive the bit error syndrome (contained in the blue box titled ``Hashing''). Based on the error syndrome, they decide whether to discard this group. Then, they proceed to information reconciliation and privacy amplification to obtain the secure key.}
\label{fig:ad_qec_frame}
\end{figure}
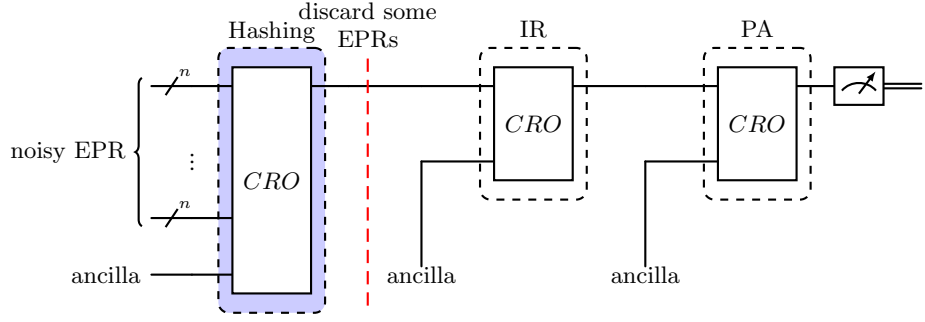

The hashing details are illustrated in Fig.~\ref{fig:hashing_detail}. Alice and Bob utilize CNOT gates within each group based on the $(n-k)\times n$ parity check matrix $H$ to operate on the noisy EPR pairs and ancillary qubits. The hashing process is a CRO, allowing the measurement to be shifted forward and rendering the preprocessing classical.

It is worth noting that the ancillary qubits can be selected as either $\ket{\psi_{00}}$ or $\ket{00}$, representing whether Alice and Bob opt to encrypt their messages using OTP during the two-way communication step. As detailed in Sec.~\ref{sec:one-way-sec}, OTP encryption decouples the bit and phase errors. Surprisingly, we show that the OTP step is unnecessary, and the key rate is \emph{higher} when Alice and Bob do not encrypt their messages, as proved in Lemma \ref{lem:phase_without_OTP}.

In the original B step, hashing is performed between data qubits, as shown in Fig.~\ref{fig:B_step}. Our framework utilizes additional ancillary qubits to carry the hashing value.  It turns out that in-place hashing yields the same key rate as using additional ancillary qubits without OTP encryption. We compare these two methods in Appendix \ref{app:key_rate}.

Our QKD advantage distillation framework based on classical linear code is summarized as follows.

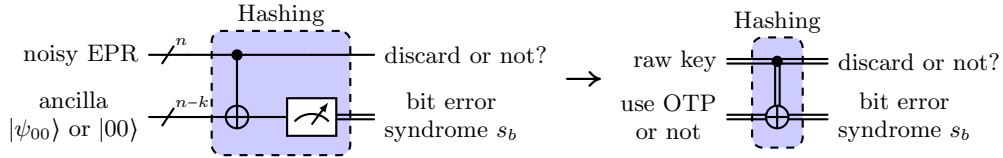
\begin{figure}[!h]
	\begin{quantikz}
\lstick[1]{noisy EPR} & \qwbundle{n} & \ctrl{1}\gategroup[2,steps=2,style={dashed,rounded corners,fill=blue!20, inner xsep=2pt},background,label style={label position=above,anchor=north,yshift=0.25cm}]{Hashing} &  & \rstick[1]{discard or not?} \\
\lstick[1]{ancilla \\ $\ket{\psi_{00}}$ or $\ket{00}$} &\qwbundle{n-k}& \targ{} & \meter{} & \setwiretype{c} \rstick[1]{bit error \\ syndrome $s_b$}
\end{quantikz} \raisebox{0.4\height}{\scalebox{2}{$\rightarrow$}} 
\begin{quantikz}
\lstick[1]{raw key}  \setwiretype{c} & \ctrl[vertical wire=c]{1}\gategroup[2,steps=1,style={dashed,rounded corners,fill=blue!20, inner xsep=2pt},background,label style={label position=above,anchor=north,yshift=0.25cm}]{Hashing} & \rstick[1]{discard or not?}\\
\lstick[1]{use OTP\\ or not} \setwiretype{c} &\targ{} & \rstick[1]{bit error \\ syndrome $s_b$}
\end{quantikz}
\caption{Preprocessing of raw keys using classical linear code. Left: Alice and Bob employ CNOT gates from noisy EPR pairs to ancillary qubits within each group based on the $(n-k)\times n$ parity check matrix $H$. Following the measurements of ancillary qubits and two-way communication (not depicted), they obtain the bit error syndrome $s_b$. The decision to discard the group is based on the key rate associated with this bit error syndrome. Right: As the hashing process is a CRO, the measurement can be shifted forward, rendering the entire process classical. The choice of ancillary qubits, either $\ket{\psi_{00}}$ or $\ket{00}$, corresponds to the use of OTP encryption or not, respectively.}
\label{fig:hashing_detail}
\end{figure}

\begin{mybox}[label={box:ad_protocol}]{Advantage distillation based on classical linear code}
	\textbf{Input:} 
		\begin{enumerate}
			\item Raw keys obtained after data sifting.
			\item {BB84 protocol: bit error $\delta_b$, phase error $\delta_p$. \\
			   Six-state protocol: Bell-diagonal terms $p_{00}, p_{01}, p_{10}, p_{11}$.}
		\end{enumerate}
	\textbf{Advantage distillation:}
		\begin{enumerate}
			\item Alice and Bob choose a $[n\ k\ d]$ code, with parity check matrix $H$ and generator matrix $G$.
			\item Alice and Bob divide $n$-bit raw key into one group.
			\item Alice and Bob implement the same hash operation $H$ to raw keys in each group.
			\item Alice and Bob compare their results via two-way communication to deduce bit error syndrome $s^j_b$ for each group. They can choose whether to perform OTP encryption.
			\item For each bit error syndrome $s_b^j$, Alice and Bob keep this group if the key rate $r^j$ corresponding to error syndrome $s_b^j$ satisfying $r^j > 0$. Otherwise, they discard the group. The key rate formula for $r^j$ is given in Sec.~\ref{sec:qkd_adv_distill_key_rate}.
			\item Alice and Bob proceed to information reconciliation and privacy amplification on the raw keys to obtain secure keys.
		\end{enumerate}
\end{mybox}

\section{Key rate formula for advantage distillation}\label{sec:qkd_adv_distill_key_rate}

Here, we derive the key rate formula for our QKD advantage distillation framework. Firstly, we introduce some notations. For a multiset $S$ of positive real numbers satisfying $\sum_{x\in S} x = 1$, the entropy of this multiset $h(S)$ is defined as:
\begin{equation}
	\begin{split}
		h(S) &= - \sum_{x \in S} x \log x.
	\end{split}
\end{equation}
The entropy of a random variable $X$ with probability mass function $p(X = x)$ is given by $h(X) = h(\{p(X=x)\}_x)$. The conditional entropy of two random variable $X,Y$ is defined as 
\begin{equation}
	h(Y|X) = \sum_x p(X=x) h(Y|X=x) = \sum_x p(X=x) h(\{p(Y=y|X=x)\}_y).
\end{equation}

 Suppose Alice and Bob choose $[n\ k\ d]$ code $C$ with parity check $H$ and generator matrix $G$. There are totally $2^{n-k}$ bit error syndromes $\{s^j_b\}_j$. For a specific syndrome, there are $2^{k}$ bit error patterns $\{e^j_{b,i}\}_i$ corresponding to this syndrome. Denote the probability of getting bit error syndrome $s^j_b$ as $q^j$, and the probability of getting bit error pattern $e^{j}_{b,i}$ as $q^j_i$. By definition, we have
\begin{equation}
	q^j = \sum_{i=0}^{2^k-1} q^j_i.
\end{equation}

We classify the phase error patterns according to the dual code $C^{\perp}$ of $C$. That is, two phase errors $e^{j'}_{p,i'}$ and $e^{j'}_{p,k'}$ corresponding to same syndrome $s_p^{j'}$ if $G^T e^{j'}_{p,i'} = G^T e^{j'}_{p,k'}$. There are $2^{k}$ phase error syndrome $\{s^{j'}_{p}\}_{j'}$ and  $2^{n-k}$ phase error patterns $\{e^{j'}_{p,i'}\}_{i'}$ corresponding to a specific syndrome $s^{j'}_{p}$. Define $q^{jj'}_i$ as the probability of getting bit error pattern $e^j_{b,i}$ and phase error syndrome $s^{j'}_{p}$, and $q^{jj'}_{ii'}$ as the probability of getting bit error pattern $e^j_{b,i}$ and phase error pattern $e^{j'}_{p,i'}$. Then, by definition, we have
\begin{equation}
\begin{split}
	q^{j}_i &= \sum_{j'=0}^{2^k-1} q^{jj'}_i,\\
	q^{jj'}_i &= \sum_{i'=0}^{2^{n-k}-1} q^{jj'}_{ii'}.
\end{split}
\end{equation}

We list our notations in the following table.

\begin{center}
\begin{table}[H]
	\caption{Notation}\label{tab:Notation} \centering
	\begin{tabular}{cc}
		\hline
		bit error syndrome & $s^j_b$\\
		bit error pattern & $e^j_{b,i}$\\
		phase error syndrome & $s^{j'}_p$\\
		phase error pattern & $e^{j'}_{p,i'}$\\
		probability for bit error syndrome $s^j_b$ & $q^j$\\
		probability for bit error pattern $e^j_{b,i}$ & $q^j_i$\\
		probability for bit error pattern $e^j_{b,i}$ and phase error syndrome $s^{j'}_p$ & $q^{jj'}_i$\\
		probability for bit error pattern $e^j_{b,i}$ and phase error pattern $e^{j'}_{p,i'}$ & $q^{jj'}_{ii'}$\\
		\hline
	\end{tabular}
\end{table}
\end{center}

\subsection{With OTP encryption}
We derive the key rate formula for advantage distillation with OTP encryption in two-way classical communication. Alice and Bob have noisy EPR pairs divided into groups, each containing $n$ EPR pairs. The initial step involves the application of the parity check matrix $H$ to identify the error syndrome $s^j$. Subsequently, OTP encryption consumes $n-k$ secret keys per group.

After identifying the bit error syndrome $s^j_b$, the bit error pattern $e^j_{b,i}$ is contained in the set $\{e^j_{b,i}\}_i$. Then, Alice and Bob perform one-way error correction to correct $e^j_{b,i}$. Based on Lemma \ref{lem:error_space}, the asymptotic secret key consumption is the entropy of bit error patterns conditional on the bit error syndrome $s^j$, which can be expressed as 
\begin{equation}\label{eq:consume_bit_pattern}
	I_{b}^j = h(\{\frac{q^j_i}{q^j}\}_{i}).
\end{equation}

Upon identification of the bit error pattern, the focus shifts to determining the phase error pattern $e^{j'}_{p,i'}$. As mentioned in Sec. \ref{sec:one-way-sec}, to guarantee the keys are secure, Alice and Bob do not need to identify the exact phase error pattern \cite{Shor_simple_2000} but only need to calculate the length of the tag required to perform phase error correction \cite{Gottesman_2004_security, Huang_stream_2022}, which gives the key consumption of privacy amplification. According to Lemma \ref{lem:error_space}, the length of the tag is given by:
\begin{equation}\label{eq:ph_consumption_OTP}
\begin{split}
	I^{j}_{p,i} &= h(\{\frac{q^{jj'}_{ii'}}{q^j_i}\}_{j',i'})	\\
	&= h(\{\frac{q^{jj'}_{i}}{q^j_i}\}_{j'}) + \sum_{j'} \frac{q^{jj'}_i}{q^j_i}h(\{\frac{q^{jj'}_{ii'}}{q^{jj'}_i}\}_{i'}),
\end{split}
\end{equation}
where the last equality is derived from the chain rule of the entropy function $h$. That is, for two random variables $X,Y$, the entropy of their joint distribution $(X,Y)$ satisfies
\begin{equation}
\begin{split}
	h((X,Y)) &= h(X) + h(Y|X) \\
	&= h(\{p(X=x)\}_x) + \sum_x p(X=x) h(\{p(Y=y|X=x)\}_y). 
\end{split}
\end{equation}

The cumulative consumption in the one-way error correction of syndrome $s^j_b$ is expressed as:
\begin{equation}\label{eq:tot_consumption_OTP}
\begin{split}
	R^j &= I^j_b + \sum_i \frac{q^j_i}{q^j} I_{p,i}^j \\
            &= h(\{\frac{q^j_i}{q^j}\}_{i}) + \sum_{i} \frac{q^j_i}{q_j} h(\{\frac{q^{jj'}_{ii'}}{q^{j}_i}\}_{j',i'})	\\
            &= h(\{\frac{q^{j,j'}_{i,i'}}{q^j}\}_{i,j',i'}).
\end{split}
\end{equation}
The first line arises from the fact that the cost of the raw keys $R^j$, after identifying the syndrome $s_b^j$, comprises two components: the cost of identifying the bit error pattern, given by \( I_b^j \), and the expected cost of identifying the phase error pattern across different bit error patterns, calculated as $\sum_i \frac{q_i^j}{q^j} I_{p,i}^j$.

The key idea for advantage distillation to increase the key rate and error threshold is as follows: If $R^j$ exceeds $n$, the corresponding group of keys does not contribute positively to the final secret key and thus will be discarded. Conversely, if $R^j$ is less than $n$, this group will be retained and produce $n - R^j$ bits of secure keys. Combining with the consumption $n-k$ of OTP encryption for each group, the key rate formula for advantage distillation with OTP encryption is obtained.

\begin{theorem}[Key rate formula for advantage distillation with OTP encryption]\label{thm:kr_OTP} 
Suppose Alice and Bob perform two-way post-processing with one-time pad encryption in the advantage distillation framework using an $[n\ k\ d]$ linear code. The probabilities of getting bit error syndrome $s^j_b$, bit error pattern $e^j_{b,i}$, phase error syndrome $s_p^{j'}$ and phase error pattern $e^{j'}_{p,i'}$ are given in Table \ref{tab:Notation}. Then, the key rate is
\begin{equation}
\begin{split}
	r &= \sum_j q^jr^j  - \frac{n-k}{n},\\
	r^j&= \max\{1-\frac{1}{n}h(\{\frac{q^{j,j'}_{i,i'}}{q^j}\}_{i,j',i'}), 0\}.
\end{split}
\end{equation}
\end{theorem}

We stress that for the six-state protocol, we can calculate the probability of error patterns based on the bit and phase error rates, thereby determining the key rate $r$. However, for the BB84 protocol, a free parameter exists. Therefore, minimizing the key rate $r$ across all valid free parameters is necessary.

To enhance the key rate formula outlined in Theorem \ref{thm:kr_OTP}, we can compress the bit error syndrome during two-way communication with OTP encryption, leveraging Lemma \ref{lem:error_space} once more. This compression reduces the consumption of secret keys for OTP to the entropy of bit error syndromes $h(\{q^j\}_j)$. Such compression enables a more efficient use of OTP, leading to an improved key rate.

\begin{theorem}[Key rate formula for advantage distillation with OTP encryption and bit error syndrome hashing]\label{thm:kr_OTP_hash}
Suppose Alice and Bob perform two-way post-processing with hashing of the tag and one-time pad encryption in the advantage distillation framework using an $[n\ k\ d]$ linear code. The probabilities of getting bit error syndrome $s^j_b$, bit error pattern $e^j_{b,i}$, phase error syndrome $s_p^{j'}$ and phase error pattern $e^{j'}_{p,i'}$ are given in Table \ref{tab:Notation}. Then, the key rate is
\begin{equation}\label{eq:key_rate_OTP_hash}
\begin{split}
	r &= \sum_j q^jr^j  - \frac{h(\{q^j\}_j)}{n},\\
	r^j&= \max\{1-\frac{1}{n}h(\{\frac{q^{j,j'}_{i,i'}}{q^j}\}_{i,j',i'}), 0\}.
\end{split}
\end{equation}
\end{theorem}

In Theorem \ref{thm:kr_OTP} and Theorem \ref{thm:kr_OTP_hash}, OTP encryption is employed during information reconciliation, consuming secret keys. However, as shown in the next section (Lemma \ref{lem:phase_without_OTP}), omitting OTP encryption enhances the key rate. Building on this observation, we further enhance the key rate in Theorem \ref{thm:kr_OTP} with the improvements outlined in Theorem \ref{thm:kr_noOTP}.

We remark that another information reconciliation method, involving directly measuring a portion of the qubits, yields a slight improvement over the results obtained from Theorem \ref{thm:kr_OTP_hash}. However, for brevity, we do not present the key rate formula in the main text. A detailed discussion is available in Appendix~\ref{app:another_IR}. 

\subsection{Without OTP encryption}\label{sec:key_rate_noOTP}
The following section illustrates how omitting OTP encryption enhances the key rate. The key concept involves understanding the impact on the phase error when Alice and Bob do not encrypt their message while identifying the bit error syndrome. This is equivalent to utilizing $\ket{00}$ as ancillary qubits instead of employing EPR pairs to gather parity information about the noisy EPR pairs. We begin by examining the case where $H$ is the hashing matrix of the $[7\ 4\ 3]$ code:
\begin{equation}
    H = \begin{pmatrix}
        1 & 0 & 0 & 1 & 1 & 0 & 1 \\
        0 & 1 & 0 & 1 & 0 & 1 & 1 \\
        0 & 0 & 1 & 0 & 1 & 1 & 1
    \end{pmatrix}.
\end{equation}

The process of hashing and measuring the bit error syndrome is depicted in Fig.~\ref{fig:phase_error}. Here, $\ket{00}$ is the ancillary qubit. CNOT gates are applied between data qubits and the ancillary qubits, according to the hashing matrix $H$. Subsequently, the ancillary qubits are measured to extract parity information. This operation does not affect the bit error pattern. We transform the quantum circuit using Hadamard gates to analyze the effect on the phase error pattern. The equivalent circuit is shown in Fig.~\ref{fig:phase_error}, where the ancillary qubits become the control qubits of CNOT gates. 
Each ancillary qubit pair, $\ket{00}$, is a superposition of Bell states: $\ket{00} = \frac{1}{\sqrt{2}} (\ket{\psi_{00}} + \ket{\psi_{01}})$. 
Since the non-Bell-diagonal terms do not affect the distillation process (see Appendix \ref{app:non-Bell-diagonal}), this pair can be treated as passing through a dephasing channel in the Bell basis, resulting in an equal mixture of Bell states $\ket{\psi_{00}}$ and $\ket{\psi_{01}}$.
If the ancillary qubit pair is in the state $\ket{\psi_{01}}$, the CNOT gate will alter the phase error pattern. 
Define a bit string $a$ of length $n-k$ with $a_i = 1$ if the $i$-th ancillary qubits are $\ket{\psi_{01}}$. The bit string $a$ follows a uniform distribution on $\{0,1\}^{n-k}$. This bit string will lead to a change in the phase error by the CNOT gate, as illustrated in Fig.~\ref{fig:phase_error}. After measuring the ancilla qubits, the phase error pattern becomes $\tilde{e}_P = e_P + H^Ta$. This relationship can be easily extended to an arbitrary linear hashing matrix, leading to the following lemma.

\begin{figure}
	\begin{quantikz}[row sep=0.3cm, column sep=0.3cm]
		\lstick[7]{noisy \\EPR pairs}& \ctrl{7} &&&&&&&  \\
		&& \ctrl{7} &&&&&&  \\
		&&& \ctrl{7} &&&&& \\
		&&&& \ctrl{5} &&&& \\
		&&&&& \ctrl{5} &&& \\
		&&&&&& \ctrl{4} &&  \\
		&&&&&&& \ctrl{3} &  \\
		\lstick{$\ket{0}$}&\targ{}&&& \targ{} &\targ{}&& \targ{} &\meter{}  \\
		\lstick{$\ket{0}$}&& \targ{} && \targ{} && \targ{} &\targ{}&\meter{}  \\
		\lstick{$\ket{0}$} &&& \targ{} && \targ{} & \targ{} & \targ{} &\meter{}  \\
	\end{quantikz}   \raisebox{0.7\height}{\scalebox{2}{$=$}} 
	\begin{quantikz}[row sep=0.3cm, column sep=0.3cm]
		\lstick[7]{noisy \\EPR pairs}\slice{(a)}&\gate{H} \slice{(b)}& \targ{} &&\slice{(c)}& \gate{H} \slice{(d)}&&  \\
		&\gate{H} && \targ{} && \gate{H} & & \\
		&\gate{H} &&& \targ{} & \gate{H} & &\\
		&\gate{H} &\targ{} & \targ{} && \gate{H} && \\
		&\gate{H} &\targ{} && \targ{} & \gate{H} & & \\
		&\gate{H} & & \targ{} & \targ{} & \gate{H} & & \\
		&\gate{H} & \targ{} &\targ{}& \targ{} & \gate{H} &&  \\
		\lstick{$\ket{0}$}& \gate{H}&\ctrl{-7}&&& \gate{H} &&\meter{}  \\
		\lstick{$\ket{0}$}& \gate{H}&& \ctrl{-7} && \gate{H} &&\meter{} \\
		\lstick{$\ket{0}$}& \gate{H}&&& \ctrl{-7} & \gate{H} &&\meter{}  
	\end{quantikz}
	\caption{Analysis of the phase error pattern when $\ket{00}$s are employed as ancillary qubits, using the $[7\ 4\ 3]$ code as an example. Left: Alice and Bob apply CNOT gates from noisy EPR pairs to ancillary qubits $\ket{00}$. After measuring the ancillary qubits and engaging in two-way communication (not depicted), they determine the bit error syndrome for this group.  Right: Equivalence of the quantum circuit. (a) The phase error pattern is $e_p$. (b) After the Hadamard transform, the phase error pattern $e_p$ transforms into the bit error pattern. (c) The $\ket{00}$ ancillary qubit pairs are the control qubits of CNOT gates to the EPR pairs. As discussed in the main text, $\ket{00}$ can be treated as an equal mixture of $\ket{\psi_{00}}$ and $\ket{\psi_{01}}$. The phase error $e_p$ is converted to $\tilde{e}_p = e_p + H^T a$, where $a$ is uniformly random on $\{0,1\}^{n-k}$. (d) The bit error pattern $\tilde{e}_p$ becomes the phase error pattern.}
	\label{fig:phase_error}
\end{figure}
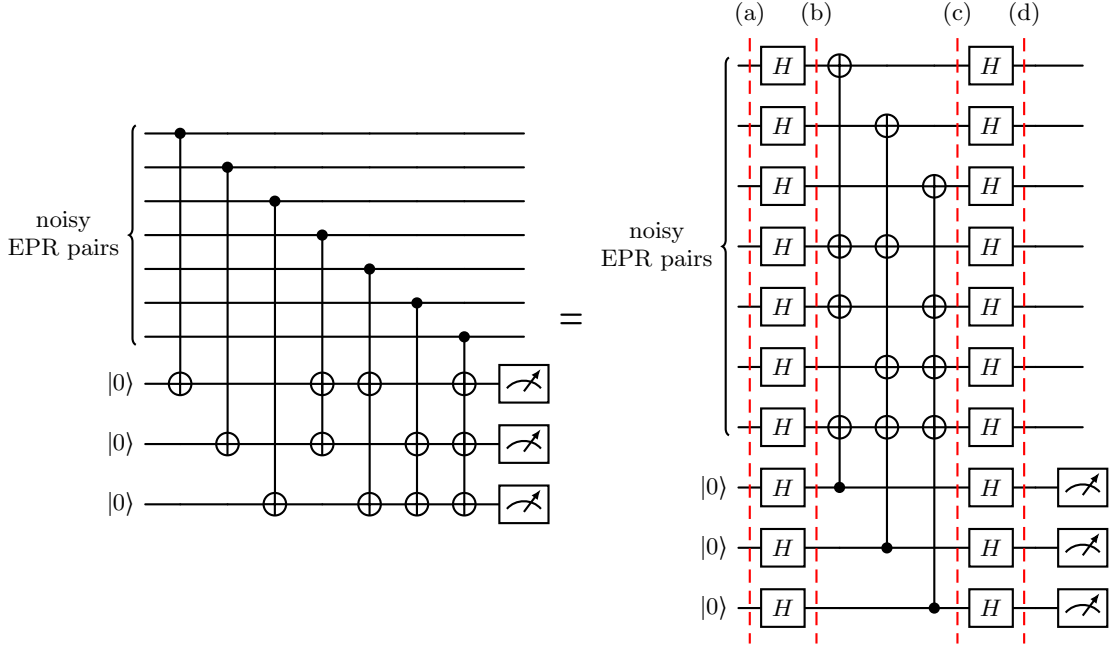

\begin{lemma}[Transformation of the phase error pattern without OTP encryption]\label{lem:phase_error_evolution}
    Assuming Alice and Bob hold $n$ noisy EPR pairs with the original phase error pattern denoted as $e_{p}$. They perform the same linear hash using an $(n-k) \times n$ parity check matrix $H$ between the EPR pairs and ancillary qubits $\ket{00}^{\otimes (n-k)}$. EPR pairs serve as control qubits, while ancillary qubits act as target qubits. If the ancillary qubits are measured, and the results are publicly announced without OTP encryption, the phase error pattern transforms to:
    \begin{equation}\label{eq:ph_err_enlarge}
    \begin{split}
    	\tilde{e}_P = e_P + H^T a,
    \end{split}
    \end{equation}
    where $a$ is uniformly distributed on $\{0,1\}^{n-k}$.
\end{lemma}

\begin{proof}	
	As discussed above, the ancillary qubits can be treated as an equal mixture of $\ket{\psi_{00}}$ and $\ket{\psi_{01}}$. These ancilla qubits can be represented as $\otimes_{i=1}^{n-k} \ket{\psi_{0,a_i}}$, where $a = (a_1, a_2, \dots, a_{n-k})$ is uniformly distributed over $\{0, 1\}^{n-k}$. 
	
	Suppose the initial phase error pattern is given by $e_p$ and $H_{ij} = 1$. In this scenario, a bilateral CNOT operation is performed between the $i$-th ancillary qubit pair and the $j$-th noisy EPR pair. 
	
	According to Eq.~\eqref{eq:bilateralCNOT}, the bilateral CNOT operation updates the phase error pattern from $e_p$ to $e_p + \hat{v}(j, a_i)$, where $\hat{v}(j, a_i)$ is a vector that has all entries set to zero except for position $j$, which is set to $a_i$. 
	
	After all the bilateral CNOT operations have been performed, the resulting phase error pattern is updated to $e_p + H^T a$.
\end{proof}

A direct implication of Eq.~\eqref{eq:ph_err_enlarge} is that the key rate without OTP is at least as large as the case with OTP, making the use of OTP encryption for bit error correction unnecessary. This observation is summarized in the following lemma.

\begin{lemma}[Unneccessity of OTP encryption]\label{lem:phase_without_OTP}
Assuming Alice and Bob hold $n$ noisy EPR pairs and perform bit error correction using a $(n-k) \times n$ linear hashing matrix $H$ with ancillary qubits $\ket{00}^{\otimes (n-k)}$, the additional consumption of perfect EPR pairs to correct phase error is no more than $(n-k)$ compared to using $\ket{\psi_{00}}^{\otimes (n-k)}$ as ancillary qubits.
\end{lemma}

\begin{proof}
	The key idea is to track the $\varepsilon$-probable set of phase error patterns and show that it expands by at most a factor of $2^{n - k}$ after hashing. According to Lemma \ref{lem:error_space}, this implies that the additional consumption of secret keys for phase error correction in each group will not exceed $n - k$.	
	
	Specifically, let $\mathcal{T}$ denote the $\varepsilon$-probable set of phase error patterns before hashing. Consider the set
	\begin{equation}
		\mathcal{T}' = \{ e + H^\top a \mid e \in \mathcal{T},\ a \in \{0,1\}^{n - k} \},
	\end{equation}
	which satisfies $|\mathcal{T}'| \le 2^{n - k} |\mathcal{T}|$.	According to Lemma \ref{lem:phase_error_evolution}, $\mathcal{T}'$ contains all possible phase error patterns after applying the hashing process, given that the original phase error pattern is in $\mathcal{T}$. From the definition of $\mathcal{T}$, it follows that
	\begin{equation}
		\Pr[e_p \in \mathcal{T}'] \ge 1 - \varepsilon.
	\end{equation}
	Let $\mathcal{T}''$ denote the $\varepsilon$-probable set of phase error patterns after hashing. Since $\mathcal{T}''$ is the smallest probable set, we have $|\mathcal{T}''| \le |\mathcal{T}'|$. By Lemma \ref{lem:error_space}, the additional consumption of secret keys for phase error correction increases by at most
	\begin{equation}
		\log |\mathcal{T}''| - \log |\mathcal{T}| \le n - k.
	\end{equation}
\end{proof}

Note that $n - k$ represents the consumption for OTP encryption. Moreover, we can take $H$ as the hashing matrix in information reconcilliation. Thus, we have proved that the additional consumption for phase error correction when performing bit error correction without OTP—whether during preprocessing or information reconciliation—does not surpass the consumption for OTP encryption. As mentioned earlier, phase error correction can be reduced to privacy amplification in QKD, where we do not need to explicitly know the phase error pattern but only need to account for the consumption required for phase error correction. Therefore, one can always perform information reconciliation without OTP and consume additional secret keys during phase error correction. That is, the key rate without using OTP is no less than the key rate with OTP and has the potential to be higher without OTP through careful analysis, especially for two-way preprocessing.

We now analyze the potential savings achievable in the scheme without OTP encryption through a detailed examination of the change in the phase error pattern. Multiplying $G^T$ on both sides of Eq.~\eqref{eq:ph_err_enlarge} and utilizing the fact that $HG = G^TH^T = 0$, we obtain:
\begin{equation}\label{eq:ph_syndrom}
	G^T \tilde{e}_p = G^T e_p.
\end{equation}
This implies that, after two-way post-processing without OTP encryption, the phase error syndrome remains the same, while the probability of obtaining each phase error pattern within the same phase error syndrome becomes equal. Denoting the probability of obtaining bit error pattern $e^{j}_{b,i}$ and phase error pattern $e^{j'}_{p,i'}$ in this scenario as $\tilde{q}^{jj'}_{ii'}$, then:
\begin{equation}
	\tilde{q}^{jj'}_{ii'} = \frac{q^{jj'}_i}{2^{n-k}}.
\end{equation}

The consumption of phase error correction becomes
\begin{equation}\label{eq:ph_consumption_noOTP}
\begin{split}
	\tilde{I}^{j}_{p,i} &= h(\{\frac{\tilde{q}^{jj'}_{ii'}}{q^j_i}\}_{j',i'})	\\
	&=h(\{q^{jj'}_{i}\}_{j'}) + (n-k).
\end{split}
\end{equation}
The total consumption of groups with bit error syndrome $s^j$ becomes:
\begin{equation}\label{eq:tot_consumption_noOTP}
\begin{split}
	\tilde{R}^j &= I_{b}^j + \sum_{i}\frac{q^j_i}{q_j}\tilde{I}_{p,i}^j \\
    &= h(\{\frac{q^j_i}{q^j}\}_{i}) + \sum_{i} \frac{q^j_i}{q^j} h(\{\frac{q^{jj'}_{i}}{q^j_i}\}_{j'})	 + n- k \\
	&= h(\{\frac{q^{jj'}_{i}}{q^j}\}_{j',i}) + n - k.
\end{split}
\end{equation}
Next, Alice and Bob discard the groups whose consumption exceeds the gain. By combining the above, we derive the following key rate formula.

\begin{theorem}[Key rate formula for advantage distillation, without OTP encryption]\label{thm:kr_noOTP}
Suppose Alice and Bob perform two-way post-processing without one-time pad encryption in the advantage distillation framework using an $[n\ k\ d]$ linear code. The probabilities of getting bit error syndrome $s^j$, bit error pattern $e^j_i$, phase error syndrome $s_P^{j'}$ and phase error pattern $e^{j'}_{P,i'}$ are given in Table \ref{tab:Notation}. Then, the key rate is
\begin{equation}
\begin{split}
	r &= \sum_j q^jr^j,\\
	r^j&= \max\{\frac{k}{n}-\frac{1}{n}h(\{\frac{q^{j,j'}_{i}}{q^j}\}_{i,j'}), 0\} .
\end{split}
\end{equation}
\end{theorem}

Now we compare Eq.~\eqref{eq:tot_consumption_noOTP} with previous consumption $R^j$, given in Eq.~\eqref{eq:tot_consumption_OTP}, plus the consumption of OTP, $O=n-k$. For the group with bit error syndrome $s^j$, the savings of EPR pairs in this group is
\begin{equation}
\begin{split}
	R^j + (n-k) - \tilde{R}^j = \sum_{i,j'} \frac{q^{jj'}_{i}}{q^j_i} h(\{\frac{q^{jj'}_{ii'}}{q^{jj'}_{i}}\}_{i'}).
\end{split}
\end{equation}
This demonstrates that advantage distillation without OTP encryption results in a higher key rate than the scheme with OTP encryption in Theorem \ref{thm:kr_OTP} \cite{oneway_footnote}. However, there is no rigorous guarantee omitting OTP encryption outperforms Theorem \ref{thm:kr_OTP_hash}. 

Lemma \ref{lem:phase_without_OTP} can be directly applied to enhance the key rate in Theorem \ref{thm:kr_OTP_hash}. Specifically, if we compress the syndrome but do not perform OTP encryption during two-way classical communication, the key rate formula will outperform Theorem \ref{thm:kr_OTP_hash}. Watanabe et al. \cite{Watanabe2007keyrate} have previously shown this scheme based on the $[2\ 1\ 2]$ code, deriving the key rate formula using a different method. We compare these approaches in Sec.~\ref{sec:key_rate_improve} and demonstrate that Theorem \ref{thm:kr_noOTP} can outperform theirs in the high error rate regime. The key rate formula for a general error-correcting code involving error syndrome compression and omitting OTP encryption is left for future work.

\subsection{Adding noise combined with classical linear code}\label{sec:add_noise}
Adding noise has been demonstrated to increase the threshold in QKD protocols \cite{PhysRevLett.95.080501, PhysRevLett.98.020502}. In this approach, Alice deliberately introduces noise to her sifted key. While increasing the key rate in the entanglement-based paradigm may seem challenging, a crucial observation is that a private state is sufficient for secret key generation \cite{PhysRevLett.94.160502}. By adding noise, Alice reduces the knowledge of Eve about the phase error pattern. Remarkably, it has been demonstrated that Alice and Bob do not need to correct all the phase error patterns, leading to savings in consumption during the phase error correction step \cite{PhysRevA.54.1869, PhysRevLett.98.020502}. This method has proven effective in simulation results, particularly for large error rates. It can extend the tolerable error rate from $11\%$ to $12.4\%$ for the BB84 protocol and from $12.7\%$ to $14.12\%$ for the six-state protocol \cite{PhysRevLett.95.080501}.

We integrate the adding noise method into the advantage distillation framework. The noise introduced in \cite{PhysRevLett.95.080501, PhysRevLett.98.020502} is $i.i.d.$ for each noisy EPR pair.  However, as the central idea of advantage distillation is identifying the bit error syndrome of each group of EPR pairs and subsequently implementing error correction and privacy amplification, the bit errors of different EPR pairs within a group are correlated. Consequently, introducing $i.i.d.$ noise would alter the bit error syndrome of each group, which may be inappropriate for this scenario. To address this, a natural solution is to introduce structured noise to preserve the bit error syndrome for each group. Thus, Alice needs to add noise patterns to the entire group instead of introducing noise to each EPR pair individually. Appendix \ref{app:add_noise} presents the detailed method for adding structured noise and subsequent analysis of the key rate based on private state distillation.

\section{Case study} \label{sec:application}
This section applies our advantage distillation framework based on classical linear code to several scenarios. Firstly, we demonstrate that the best-known result for tolerable error rate can be directly derived by combining the $[n\ 1\ n]$ code with Theorem \ref{thm:kr_noOTP}. Secondly, we enhance the state-of-the-art key rate in the high error rate regime by utilizing $[n\ 1\ n]$ codes and $[m\ m-1\ 2]$ codes.

\subsection{Derive the threshold using $[n\ 1\ n]$ code}\label{sec:threshold}
We show how our framework offers a systematic approach to determining the tolerable error rate of various QKD protocols. Previously, the Shor-Preskill formula \cite{Shor_simple_2000} establishes an $11\%$ threshold for the BB84 protocol, which is subsequently improved to $12.7\%$ for the six-state protocol by Lo \cite{Lo_proof_2001}. Gottesman and Lo further enhance the BB84 protocol to tolerate up to $17.9\%$, and the six-state protocol to $26.4\%$, by iteratively applying two-way post-processing in both bit error and phase error correction \cite{Gottesman_proof_2003}. Chau improves the threshold to $20\%$ for BB84 and $27.6\%$ for the six-state protocol through two-way post-processing coupled with an adaptive privacy amplification method \cite{Chau_2002_practical}.

In our framework, the best known thresholds—$20\%$ for BB84 and $27.6\%$ for the six-state protocol, as reported in \cite{Chau_2002_practical}—can be achieved through a single round of two-way post-processing using an \([n, 1, n]\) code, as detailed in Appendix \ref{app:thres_n1n}. While the possibility of further improvements remains open, we conjecture that the $[n\ 1\ n]$ code offers the best threshold within our advantage distillation framework based on classical linear code, owing to its maximum code distance. Potential explorations include trying adaptive coding methods in multiple rounds of preprocessing or combining with the adding noise method. These potential enhancements are left for future research.

\subsection{Enhance the key rates using $[n\ 1\ n]$ and $[m\ m-1\ 2]$ codes}\label{sec:key_rate_improve}
In this section, we apply our advantage distillation framework with $[n\ 1\ n]$ and $[m\ m-1\ 2]$ codes, where $n$ and $m$ are less than 9. The protocol steps are outlined in Box.~\ref{box:ad_protocol}. The $[n\ 1\ n]$ code is a natural generalization of the B step \cite{Gottesman_proof_2003, Ma_decoy_2006}. Taking the $[n\ 1\ n]$ code as an example, Alice and Bob first divide the raw key bits into groups, each containing $n$ noisy raw key bits. Subsequently, they apply the $[n\ 1\ n]$ code to obtain the bit error syndrome, deciding whether to discard the group based on the key rate of the syndrome provided in Theorem \ref{thm:kr_noOTP}. Then, information reconciliation and privacy amplification are performed on the remaining raw key bits.

We compare the key rate when both OTP and hashing are applied (Theorem \ref{thm:kr_OTP_hash}) with the key rate when OTP is omitted (Theorem \ref{thm:kr_noOTP}), as illustrated in Fig.~\ref{fig:compare_noOTP_OPT_hash}. The results indicate that hashing the syndrome with OTP encryption is more effective in the low error rate regime. This is reasonable because, in this regime, the entropy of the error syndrome is low, and only a few groups will be discarded. Thus, the advantage of discarding groups is limited. In contrast, the scheme without OTP performs better in the high error rate regime.

\begin{figure}
  \begin{minipage}{0.48\textwidth}
    \centering
    \includegraphics[width=\linewidth]{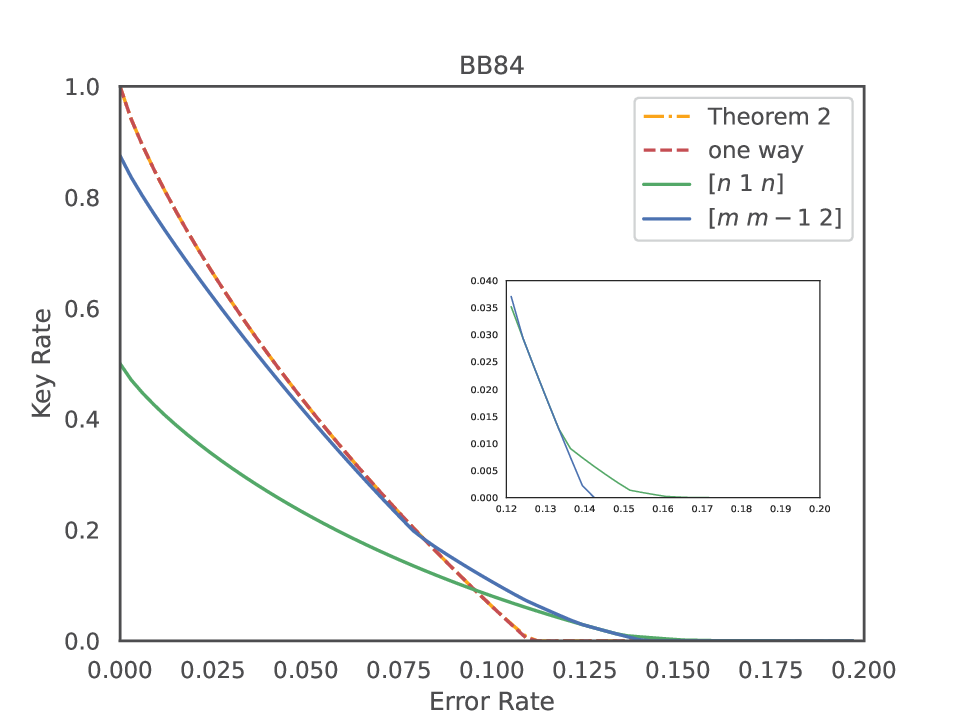}
  \end{minipage}\hfill
  \begin{minipage}{0.48\textwidth}
    \centering
    \includegraphics[width=\linewidth]{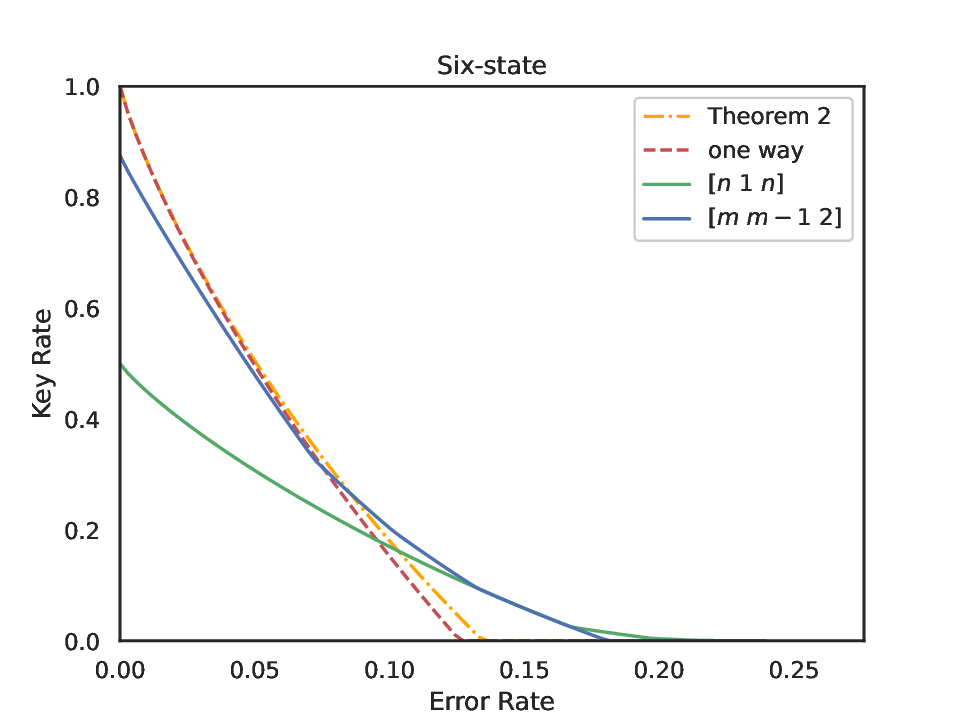}
  \end{minipage}
\caption{Comparison of the key rate when both OTP and hashing are applied (orange line) with the key rate when OTP is omitted (green line and blue line). The key rate is obtained by considering the best results from $[n\ 1\ n]$ codes and $[m\ m-1\ 2]$ codes, where $n$ and $m$ are less than 9. In the low error rate regime, hashing the syndrome with OTP encryption is more effective. Conversely, the scheme without OTP encryption performs better in the high bit error rate regime. In all regimes, our key rate surpasses the one-way protocol composed of information reconciliation and privacy amplification only.}
  \label{fig:compare_noOTP_OPT_hash}
 \end{figure}

Next, we compare our method with previous work. As discussed in Sec.~\ref{sec:key_rate_noOTP}, Lemma \ref{lem:phase_without_OTP} can be directly applied to enhance the key rate in Theorem \ref{thm:kr_OTP_hash}, showing that communication without OTP consistently outperforms communication with OTP. Watanabe et al. \cite{Watanabe2007keyrate} combined this approach with a $[2\ 1\ 2]$ code, surpassing previous results and achieving state-of-the-art performance in low error rate regimes. 
Moreover, in high error rate regimes, the method in \cite{Renner_security_2008} exhibits currently the highest key rate, as far as we know.  Note that the method in \cite{Renner_security_2008} is incorporated into the $[n\ 1\ n]$ codes: In the preprocessing step of these schemes, Alice chooses a block length $n$ and uses a private random bit $c \in \{0,1\}$ to encrypt the raw key bits $(x_1, x_2, \dots, x_n)$ in a group, resulting in $(x_1 \oplus c, x_2 \oplus c, \dots, x_n \oplus c)$, which are sent to Bob. This is equivalent to sending $(x_1 \oplus c, x_1 \oplus x_2, x_2 \oplus x_3, \dots, x_{n-1} \oplus x_n)$. Since $x_1 \oplus c$ is completely random to Bob, it carries no useful information, and the relevant message becomes $(x_1 \oplus x_2, x_2 \oplus x_3, \dots, x_{n-1} \oplus x_n)$. This is exactly the tag $s_x$ Alice sends in the $[n\ 1\ n]$ code without OTP for advantage distillation. Subsequently, Bob calculates the bit error syndrome $s_b = s_x \oplus s_y$ and checks whether $s_b = 0$ to determine whether to retain this group. Thus, the method in \cite{Renner_security_2008} corresponds to the $[n\ 1\ n]$ code with a specific syndrome $s_b = (0,0,\cdots,0)$.

Our method further enhances the outcomes presented by  \cite{Watanabe2007keyrate}, particularly in the high error rate regime. 
Our result shows that the $[m\ m-1\ 2]$ code can enhance both methods in certain noise regimes (Fig.~\ref{fig:compare_PRA07}).
Table \ref{tab:optimal_code} presents the optimal codes corresponding to different error rate regimes for the BB84 protocol and the six-state protocol, respectively.

\begin{figure}[!h]
	\begin{minipage}{0.48\textwidth}
		\centering
		\includegraphics[width=\linewidth]{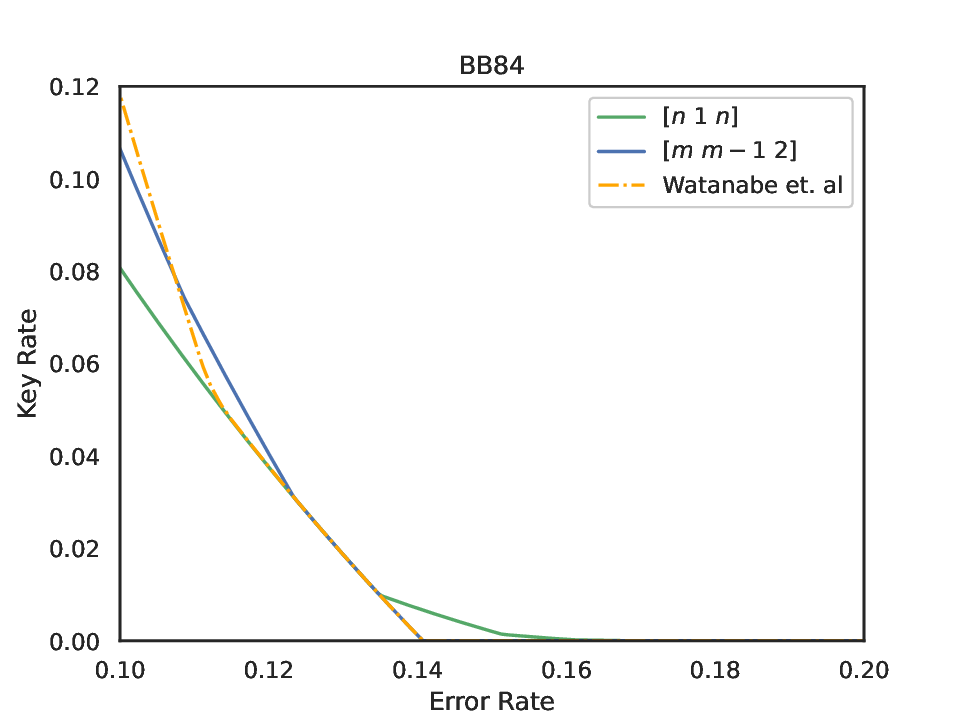}
	\end{minipage}\hfill
	\begin{minipage}{0.48\textwidth}
		\centering
		\includegraphics[width=\linewidth]{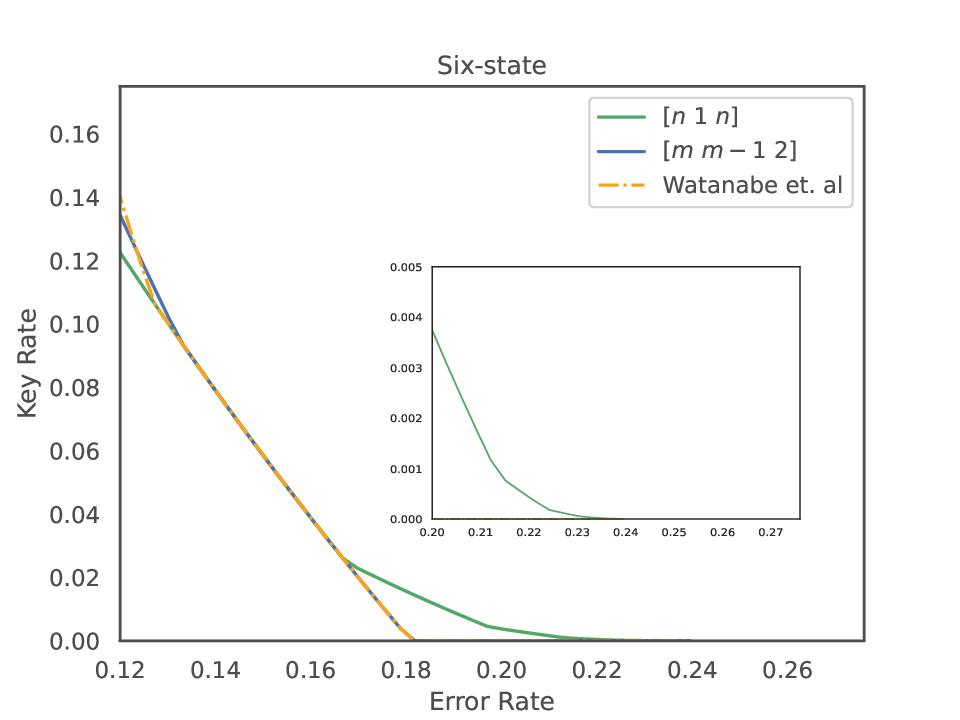}
	\end{minipage}
	\caption{Comparison of the key rate of advantage distillation without OTP (green line and blue line) and the method in \cite{Watanabe2007keyrate} (yellow dotted line). The key rate is obtained by considering the best results from $[n\ 1\ n]$ codes and $[m\ m-1\ 2]$ codes, where $n$ and $m$ are less than 9. The method in \cite{Watanabe2007keyrate} is $[2\ 1\ 2]$ code in advantage distillation framework combined with hashing the syndrome and without using OTP. The results demonstrate that $[m\ m-1\ 2]$ code (specifically, $[3\ 2\ 2]$ code) and $[n\ 1\ n]$ codes for $n \ge 3$ without hashing can improve the key rate of \cite{Watanabe2007keyrate} in the high error rate regime. The optimal codes depicted are given in Table \ref{tab:optimal_code}. }
	\label{fig:compare_PRA07}
\end{figure}
	
\begin{table}[!h] 
	\begin{minipage}{0.45\textwidth}
		\centering
		\begin{tabular}{cc}
			\hline
			Error rate & Optimal code\\\hline
			$10.9\% \sim 12.3\%$ & $[3\ 2\ 2]$ \\\hline
			$13.5\% \sim 15.1\%$ & $[3\ 1\ 3]$ \\\hline
			$15.1\% \sim 16.1\%$ & $[4\ 1\ 4]$ \\\hline
			$16.1\% \sim 16.7\%$ & $[5\ 1\ 5]$ \\\hline
			$16.7\% \sim 17.1\%$ & $[6\ 1\ 6]$ \\\hline
		\end{tabular}
	\end{minipage}
	\hspace{0.05\textwidth}
	\begin{minipage}{0.45\textwidth}
		\centering
		\begin{tabular}{cc}
			\hline
			Error rate & Optimal code\\\hline
			$12.4\% \sim 13.2\%$ & $[3\ 2\ 2]$ \\\hline
			$16.8\% \sim 19.7\%$ & $[3\ 1\ 3]$ \\\hline
			$19.7\% \sim 21.3\%$ & $[4\ 1\ 4]$ \\\hline
			$21.3\% \sim 22.4\%$ & $[5\ 1\ 5]$ \\\hline
			$22.4\% \sim 23.1\%$ & $[6\ 1\ 6]$ \\\hline
		\end{tabular}
	\end{minipage}
	
	\caption{Optimal codes for the BB84 protocol (left) and the six state protocol (right) in Fig.~\ref{fig:compare_PRA07}.}\label{tab:optimal_code} 
\end{table}

It is worth noting that the method in \cite{Watanabe2007keyrate} corresponds to a specific scenario in our framework: the $[2\ 1\ 2]$ code with syndrome compression and without OTP encryption. 
Our results suggest that $[3\ 2\ 2]$ code and $[n\ 1\ n]$ codes for $n \geq 3$ without hashing can improve the key rate compared to the $[2\ 1\ 2]$ code combined with hashing. 
We anticipate further improvement in the key rate by compressing the syndrome without OTP, which has the potential to surpass the results of \cite{Watanabe2007keyrate} in all error rate regimes by leveraging different error-correcting codes. We leave this for future work.

\section{Conclusion and outlook} \label{sec:conclusion}

In this study, we have formulated a comprehensive framework for QKD advantage distillation, generalizing previous work and proposing a novel approach based on classical linear code. Our derivations encompass key rate formulas across various scenarios, highlighting that omitting OTP encryption consistently yields a higher key rate than OTP encryption scenarios. The practical application of our framework involves utilizing $[n\ 1\ n]$ codes to derive the tolerable error rate, which aligns with the current state-of-the-art. Additionally, we employ $[n\ 1\ n]$ and $[m\ m-1\ 2]$ codes to enhance the key rate compared to prior findings systematically. 

Moreover, we extend our investigation by integrating the adding noise method into our framework, resulting in an enhanced key rate formula. This collective effort contributes valuable insights into optimizing QKD protocols and advancing the field.

We anticipate implementing our advantage distillation method across a broader range of QKD protocols. Our method can be directly applied to protocols that utilize entanglement distillation-based security proofs, such as measurement-device-independent QKD protocols \cite{Lo_measurement_2012}. Additionally, for practical QKD systems that typically employ coherent states, our method can be integrated with the decoy-state technique to enhance performance in high-noise regimes \cite{Ma_decoy_2006, Li2022improving, Huang2023sourcereplacement}, thereby extending the maximum secure transmission distance of practical QKD systems. Developing the advantage distillation framework with finite-size analysis is crucial for practical implementation and warrants further investigation. Finally, although this study focuses on discrete-variable quantum key distribution, there are many opportunities to apply our method to enhance the performance of continuous-variable quantum key distribution \cite{Jouguet_2013, Qi_2015_generating}.

There are numerous ways that are promising for further improving our methods. Firstly, while our exploration has focused on some specific codes, there is significant potential to investigate the efficiency of more sophisticated codes. By carefully designing codes, we may enhance threshold values and key rates beyond those achieved by existing protocols. Secondly, opportunities for improvement lie in synergies with additional methods, such as adaptive coding \cite{Chau_2002_practical}. Thirdly, we envisage enhancing Theorem \ref{thm:kr_OTP_hash} by considering scenarios where OTP encryption is not utilized, refining our understanding of key rates under varied conditions. 

Furthermore, while Fig.~\ref{fig:QKDflowchart} encapsulates most post-processing schemes, deviations from this structure are possible, such as altering the order of privacy amplification and information reconciliation \cite{Huang_stream_2022}. Moreover, investigating the implications of our work on the relation between classical and quantum key distillation via CROs is an intriguing direction for future research \cite{christandl2007unifyingclassicalquantumkey}. By enhancing the performance of QKD through advantage distillation protocols, we anticipate advancing both the practicality and security of QKD systems, thereby fostering their increased adoption in real-world applications.

\begin{acknowledgments}
The authors acknowledge Yizhi Huang and Pei Zeng for their helpful discussions. This work was supported by the National Natural Science Foundation of China Grants No.~12174216 and the Innovation Program for Quantum Science and Technology Grant No.~2021ZD0300804 and No.~2021ZD0300702.
\end{acknowledgments}

\appendix

\section{Entanglement distillation for non-Bell-diagonal states} \label{app:non-Bell-diagonal}
In the main text, we focus on entanglement distillation and key distillation protocols for Bell-diagonal noisy EPR states. In this section, we extend our analysis to non-Bell-diagonal states and demonstrate that the same protocols remain effective. Specifically, we prove that the entire entanglement distillation process depends solely on the Bell-diagonal components of the quantum state. The effectiveness of key distillation follows directly from our discussion in Section \ref{sec:prelim}.

For a two-qubit states $\rho$, which can be represented in the Bell basis as
\begin{equation}
	\rho = \begin{pmatrix}
		p_{00} & \times & \times & \times \\
		\times & p_{01} & \times & \times \\ 
		\times & \times & p_{10} & \times \\
		\times & \times & \times & p_{11}
	\end{pmatrix}, 
\end{equation}
where the diagonal elements $p_{00}, p_{01}, p_{10}, p_{11}$ represent the probabilities associated with the respective Bell basis states. The off-diagonal elements, denoted by $\times$, capture the coherence between these states. The Bell-diagonal part of $\rho$, denoted by $\Delta(\rho)$, as the state obtained after dephasing $\rho$ under Bell basis,
\begin{equation}
	\Delta(\rho)= \begin{pmatrix}
		p_{00} & 0 & 0 & 0 \\
		0 & p_{01} & 0 & 0 \\ 
		0 & 0 & p_{10} & 0 \\
		0 & 0 & 0 & p_{11}
	\end{pmatrix}.
\end{equation}

Denote the entire state in the entanglement distillation process, which includes both data and ancillary qubits, as $\sigma^N$. The proof proceeds in two steps.
Firstly, we show that the Bell-diagonal terms of $\sigma^N$ after applying any unitary operation $U$ within the entanglement distillation protocols in our work depend solely on the Bell-diagonal terms prior to the operation. Formally, this is expressed as
\begin{equation}\label{eq:Bell_diag_unitary}
	\Delta^{\otimes N}(U\sigma^NU^{\dagger}) = U \Delta^{\otimes N}(\sigma^N) U^{\dagger}.
\end{equation}
To establish this, observe that the unitary operations involved are limited to bilateral CNOT, $\CNOT_A \otimes \CNOT_B$, bilateral Hadamard gates $H_A \otimes H_B$, and $X_B$ in error correction, which act as permutations on the Bell states \cite{Bennett1996mixedstate}. Thus, the off-diagonal terms do not affect the diagonal terms, thereby proving Eq.~\eqref{eq:Bell_diag_unitary}.

Secondly, we show that the parity comparison and post-selection steps in both two-way and one-way distillation only relate to the diagonal terms. This follows from the properties of classically replaceable operations \cite{Liu2022classically}. Specifically, the parity comparison involves Alice and Bob performing $Z$-basis measurements on the noisy EPR pairs, followed by a classical XOR operation to determine the parity, as illustrated in Fig.~\ref{fig:parity_measurement}.
\begin{figure}[!h]
	\begin{quantikz}
		\lstick[1]{Alice} & \meter{} &  \setwiretype{c} & \ctrl[vertical wire=c]{1} & \\
		\lstick[1]{Bob}& \meter{} & \setwiretype{c} &\targ{} & \rstick[1]{parity}
	\end{quantikz} \scalebox{2}{$\rightarrow$}
	\begin{quantikz}
		\lstick[1]{Alice} & \ctrl{1} &  \meter{} & \setwiretype{c}  \\
		\lstick[1]{Bob} & \targ{} & \meter{} & \setwiretype{c} \rstick[1]{parity}
	\end{quantikz}
	\caption{Parity check process. A classical parity check after computational basis measurements is equivalent to first applying a quantum CNOT operation and then measuring the second qubit. The measurement on the first qubit does not affect the distillation process and can be delayed until the entire entanglement distillation protocol succeeds.}
	
	\label{fig:parity_measurement}
\end{figure}
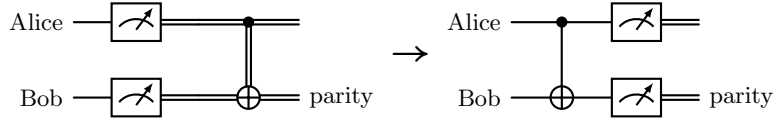

This operation is equivalent to first performing a CNOT gate followed by a measurement in the computational basis. Notably, only the measurement result of the target qubit affects the subsequent distillation process, allowing the measurement of the source qubit to be delayed. The measurement on the target qubit involves a CNOT operation followed by a computational basis measurement, which is equivalent to measuring the observable $Z_1Z_2$. This observable is diagonal in the Bell basis. Let $\Pi_{0}$ denote the eigenspace of $Z_1Z_2$ corresponding to the $+1$ eigenvalue, and $\Pi_{1}$ denote the eigenspace corresponding to the $-1$ eigenvalue. The probability of obtaining outcome $i$ is given by
\begin{equation}
	\tr(\sigma^N \Pi_i) = \tr(\Delta^{\otimes N}(\sigma^N) \Pi_i).
\end{equation}
Furthermore, the post-measurement state of the remaining qubits retains the same Bell-diagonal terms:
\begin{equation}
	\Delta^{\otimes N}\left(\tr_{AB}(\sigma^N \Pi_i)\right) = \tr_{AB}\left(\Delta^{\otimes N}(\sigma^N) \Pi_i\right),
\end{equation}
where $AB$ represents the two-qubit system depicted in Fig.~\ref{fig:parity_measurement} for the parity check.

Combining the above two steps, we demonstrate that non-Bell-diagonal states consistently preserve their Bell-diagonal terms throughout the entire distillation process. The delayed measurement of the source qubit during the parity check does not affect the distillation outcome, as it has already succeeded. Consequently, both entanglement distillation and key distillation protocols remain effective for non-Bell-diagonal states.

\section{Key rate formula for in-place hashing}\label{app:key_rate}

In the original two-way classical post-processing, hashing is performed without ancilla qubits, constituting an in-place hashing method (see Fig.~\ref{fig:B_step}). Conversely, linear hashing is performed with ancillary qubits in our advantage distillation framework based on classical linear code. It is natural to ask whether this in-place hashing yields a higher key rate. Here, we derive the key rate formula for in-place hashing and demonstrate that both methods produce the same key rate without OTP encryption.

For convenience, we assume the parity check matrix $H$ and generator matrix $G$ are in the following canonical form:
\begin{equation}
\begin{split}
	H &= \begin{bmatrix}
		A & I_{n-k}
	\end{bmatrix}, \\
	G &= \begin{bmatrix}
		I_k \\ A
	\end{bmatrix},
\end{split}
\end{equation}
where $A$ is a $(n-k) \times k$ matrix. Suppose the bit error pattern is $e_b$ and the phase error pattern is $e_p$. We write $e_b$ and $e_p$ as the concatenation of two string of length $k$ and $n-k$: 
\begin{equation}\label{eq:inplace_pattern}
\begin{split}
	e_b &= (e_{b,1}, e_{b,2}),\\
	e_p &= (e_{p,1}, e_{p,2}).
\end{split}
\end{equation}
Then, the bit error syndrome and phase error syndrome can be expressed as: 
\begin{equation}\label{eq:in_place_error}
\begin{split}
	s_b &= Ae_{b,1} + e_{b,2}, \\
	s_p &= e_{p,1} + A^T e_{p,2}.
\end{split}
\end{equation}

Upon knowing $s_b$, there is a one-to-one correspondence between $e_b$ and $e_{b,1}$. Therefore, the consumption of identifying bit error pattern $e_{b,1}$ for the previous $k$ qubits is the same as identifying $e_b$, which is given in Eq.~\eqref{eq:consume_bit_pattern}.

Next, we delve into the consumption of phase error correction. According to Eq.~\eqref{eq:ph_err_enlarge}, the phase error pattern of previous $k$ qubits is affected by the phase error pattern of later $n-k$ qubits by 
\begin{equation}
	\tilde{e}_{p,1} = e_{p,1} + A^Te_{p,2}.
\end{equation}

Comparing with Eq.~\eqref{eq:in_place_error}, we have $\tilde{e}_{p,1} = s_p$. Consequently, the consumption of phase error correction aligns precisely with identifying the phase error syndrome:
\begin{equation}
\begin{split}
	I^j_{p,i} = h(\{\frac{q^{jj'}_i}{q^{j}_i}\}_{j'}).\\
\end{split}
\end{equation}

Combining the consumption of bit and phase error correction, we derive the consumption for the groups with bit error syndrome $s^j$:
\begin{equation}
\begin{split}
	R^j &= h(\{\frac{q^j_i}{q^j}\}_i) + \sum_{i} \frac{q^j_i}{q^j} h(\{\frac{q^{jj'}_i}{q^{j}_i}\}_{j'}) \\
 		&= h(\{\frac{q^{jj'}_i}{q^j}\}_{i,j'}).
\end{split}
\end{equation}

Considering the probability of each bit error syndrome, we get the key rate formula for in-place hashing:
\begin{equation}
\begin{split}
	r &= \sum_j q^j r^j, \\
	r^j &= \max(\frac{k}{n} - \frac{1}{n}h(\{\frac{q^{jj'}_i}{q^j}\}_{i,j'}),0).
\end{split}
\end{equation}
This formula is the same as the one presented in Theorem \ref{thm:kr_noOTP}.

\section{Parity check with OTP} \label{app:another_IR}
We introduce an alternative parity check method for information reconciliation. First, Alice and Bob hash their raw keys to obtain parities. Alice then compresses this parity information and securely transmits it to Bob using OTP, enabling the derivation of the bit error syndrome. Based on this syndrome, they proceed with further key distillation steps, which may include disclosing specific bit values and discarding parts of the raw keys. For instance, in the B-step scenario, if the error syndrome is 1, Alice and Bob each reveal one bit value while retaining the remaining bits for subsequent privacy amplification. 

In a simplified scenario, we assume that either no raw key bits are announced, or the exact bit-error locations are known after announcement. Although partial key announcements and further reconciliation could be considered, this intermediate case is omitted here for simplicity. We derive the key rate formula for this method and demonstrate a marginal improvement over the approach detailed in Theorem \ref{thm:kr_OTP_hash}.

When OTP encryption is applied, the phase error in each qubit is independent.  The conditional probability of a phase error, based on whether a bit error occurs, is
\begin{equation}\label{eq:bp_cond_prob}
	\begin{split}
		\delta_{p0} = \frac{q_{01}}{q_{01}+q_{00}}, \\
		\delta_{p1} = \frac{q_{11}}{q_{10}+q_{11}}.
	\end{split}
\end{equation}
The analysis of the error patterns and syndromes aligns with Appendix \ref{app:key_rate}. Recall from Eq.~\eqref{eq:inplace_pattern} that the bit error pattern is denoted as $e_{b,i}^j = (e_{b,1},e_{b,2})^j_i$. Alice and Bob correct bit errors by measuring the first $k$ qubits and announcing the results to obtain $e_{b,1}$, from which we can deduce $e_b$, given the one-to-one correspondence between $e_b$ and $e_{b,1}$ upon knowing $s_b$. The raw key consumption in this announcement is given by
\begin{equation}
	\tilde{I}_b^j = \frac{k}{n}.
\end{equation}

Upon knowing the bit error pattern $e^j_{b,i}$, Alice and Bob can determine the phase-error probability in each qubit. They can then correct the phase errors using Lemma \ref{lem:error_space}. Denote $l_{i}^j = |e_{b,2}|$, then the phase error consumption is given by
\begin{equation}
	\tilde{I}_{p,i}^j =\frac{n-k-l_i^j}{n}h(\delta_{p0}) + \frac{l_i^j}{n} h(\delta_{p1}).
\end{equation}
and the key rate for bit error syndrome $s^j_b$ is given by 
\begin{equation}
\begin{split}
	r^j &= \max(1 - \tilde{I}_b^j -  \sum_i \frac{q^j_i}{q^j}  \tilde{I}_{p,i}^j,0) \\
	&= \max(\frac{n-k}{n} -  \sum_i \frac{q^j_i}{q^j}  \tilde{I}_{p,i}^j,0).
\end{split}
\end{equation}

Combining Theorem \ref{thm:kr_OTP_hash}, the key rate is obtained by 
\begin{equation}
\begin{split}
	r &= \sum_j q^j r^j  - \frac{h(\{q^j\}_j)}{n},\\
	r^j &= \max\{1-\frac{1}{n}h(\{\frac{q^{j,j'}_{i,i'}}{q^j}\}_{i,j',i'}), \frac{n-k}{n} -  \sum_i \frac{q^j_i}{q^j}  \tilde{I}_{p,i}^j,0 \}.
\end{split}
\end{equation}
This key rate formula offers a slight improvement over Theorem \ref{thm:kr_OTP_hash} \cite{Ma_decoy_2006, Watanabe2007keyrate}, though it does not surpass the results without OTP.

\section{Tolerable error rate for the $[n,1,n]$ code}\label{app:thres_n1n}
In Sec.~\ref{sec:threshold}, we discuss previous results on the tolerable error rate for BB84 and the six-state protocol. Here, we show how to derive these results using our framework combined with $[n\ 1\ n]$ code. For simplicity, we consider the symmetric case. In BB84, the bit error rate equals the phase error rate, i.e., \(\delta_b = \delta_p = q\), and \(\rho = \text{diag}(1-2q+q_{11},q-q_{11},q-q_{11},q_{11})\) with the free parameter \(q_{11}\). For the six-state protocol, $\rho = (1-\frac{3q}{2},\frac{q}{2},\frac{q}{2},\frac{q}{2})$. The conditional probability $\delta_{p,0}, \delta_{p,1}$ of a phase error, based on whether a bit error occurs, is given by Eq.~\eqref{eq:bp_cond_prob}.

According to Theorem \ref{thm:kr_noOTP}, if a bit error syndrome $e^j$ satisfies $r^j > 0$, then a positive key rate is achieved. For $[n\ 1\ n]$ code, $G = (1,1,\cdots,1)^T$. We consider the bit error syndrome $s^0 = (0,0,\cdots,0)$ of length $n-1$. The bit error patterns corresponding to this syndrome are in the code space:
\begin{equation}
\begin{split}
    e^0_{0} = (0,0,\cdots,0),\\
    e^0_1 = (1,1,\cdots,1).
\end{split}    
\end{equation}
and the probabilities for each bit error pattern are
\begin{equation}
\begin{split}
	q^0_0 &= (1-q)^n,\\
	q^0_1 &= q^n.
\end{split}
\end{equation}

There are only two phase error syndromes; one corresponds to phase error patterns with odd parity, and the other corresponds to patterns with even parity. Given the bit error pattern $e^0_0$, the probability of getting even-parity phase error pattern is:
\begin{equation}
\begin{split}
	\frac{q^{00}_0}{q^0_0} &= \sum_{r\ \text{is even}, 0 \le r \le n} \binom{n}{r} \delta_{p0}^{r}(1-\delta_{p0})^{n-r}\\
	&= \sum_{r=0}^n \frac{1 + (-1)^r}{2} \binom{n}{r} \delta_{p0}^{r}(1-\delta_{p0})^{n-r}\\
						   &=\frac{1+(1-2\delta_{p0})^n}{2} .
\end{split}
\end{equation}
The last equality uses the fact that $(1+x)^m = \sum_{i=0}^m \binom{m}{i}x^i$. Similarly, we have
\begin{equation}
	\frac{q^{00}_1}{q^0_1} = \frac{1+(1-2\delta_{p1})^n}{2}.
\end{equation}

From Theorem \ref{thm:kr_noOTP}, we have
\begin{equation}\label{eq:n1n_Rj}
\begin{split}
	R^0 &= nr^0 \\
		&= 1 - h(\frac{q^n}{(1-q)^n+q^n}) - \frac{(1-q)^n}{(1-q)^n+q^n}h(\frac{1+(1-2\delta_{p0})^n}{2}) - \frac{q^n}{(1-q)^n+q^n}h(\frac{1+(1-2\delta_{p1})^n}{2}) \\
	&\ge 1 - h(\frac{q^n}{(1-q)^n+q^n}) - h(\frac{1}{2} + \frac{(1-3q+2q_{11})^n + |q-2q_{11}|^n}{2((1-q)^n+q^n)}),
\end{split}
\end{equation}
where $h(x) = -x \log x - (1-x) \log (1-x)$. The last equation is from the concavity of $h$ and the fact that for the positive secret key, we must have $q < \frac{1}{3}$.

For the BB84 protocol, the lower bound of $R^0$ can be attained by setting the free parameter $q_{11}$ to zero, which minimizes the right-hand side of Eq.~\eqref{eq:n1n_Rj} from basic calculus. Thus,
\begin{equation}
\begin{split}
	R^0 &\ge 1 - h(\frac{q^n}{(1-q)^n+q^n}) - h(\frac{1}{2} + \frac{(1-3q)^n + q^n}{2((1-q)^n+q^n)})\\
	&= 1-h(\frac{x^n}{1+x^n}) - h(\frac{1}{2} + \frac{x^n + (1-2x)^n}{2(1+x^n)}),
\end{split}
\end{equation} 
where $x = \frac{q}{1-q}$. We have:
\begin{equation}
\begin{split}
	h(\frac{x^n}{1+x^n}) &= \log (1+x^n) - \frac{nx^n}{1+x^n} \log x \\
	&= \frac{1}{\ln 2}x^n - \frac{nx^n}{1+x^n} \log x + O(x^{2n}), \\
	h(\frac{1}{2} + \frac{x^n + (1-2x)^n}{2(1+x^n)}) &= 1- \frac{2}{\ln 2}(\frac{x^n+(1-2x)^n}{2(1+x^n)})^2 + O((1-2x)^{3n}).
\end{split}
\end{equation}
The first equation comes from direct calculation, and the second expands $h(x)$ at $x=\frac{1}{2}$. Thus,
\begin{equation}\label{eq:n1n_bound_rj}
	R^0 = \frac{2}{\ln 2}(\frac{x^n+(1-2x)^n}{2(x^n+1)})^2 - \frac{1}{\ln 2}x^n + \frac{nx^n}{1+x^n} \log x + O(x^{2n} + (1-2x)^{3n}).
\end{equation}

When $(1-2x)^2 > x$, the net gain of key $R^0$ becomes positive for sufficiently large $n$. This condition corresponds to $x < \frac{1}{4}$ or equivalently $q < \frac{1}{5}$. Therefore, the error rate threshold for the BB84 protocol in this scheme is $20\%$.

For the six-state protocol, $q_{11} = \frac{q}{2}$. From Eq.~\eqref{eq:n1n_Rj}, we have
\begin{equation}
\begin{split}
	R^0 &\ge 1 - h(\frac{q^n}{(1-q)^n+q^n}) - h(\frac{1}{2} + \frac{(1-2q)^n}{2((1-q)^n+q^n)}) \\
		&= 1 - h(\frac{x^n}{1+x^n}) - h(\frac{1}{2} + \frac{(1-x)^n}{2(1+x^n)}).
\end{split}
\end{equation}
And we expand $h$ at $\frac{1}{2}$:
\begin{equation}
	h(\frac{1}{2} + \frac{(1-x)^n}{2(1+x^n)}) = 1- \frac{2}{\ln 2}(\frac{(1-x)^n}{2(1+x^n)})^2 + O((1-x)^{3n}).
\end{equation}
Similarly as in Eq.~\eqref{eq:n1n_bound_rj}, when $(1-x)^2 > x$, $R^0 > 0$ for sufficiently large $n$. This condition corresponds to $x < \frac{3-\sqrt{5}}{2}$ or equivalently $q < \frac{3-\sqrt{5}}{5-\sqrt{5}} \approx 27.6\%$. Thus, the six-state protocol has an error rate threshold of $27.6 \%$ in this scheme.

\section{Key rate formula for adding noise} \label{app:add_noise}
Here, we derive the key rate formula for adding noise in the advantage distillation framework. After Alice and Bob identify the bit error syndrome $s^j$ with OTP encryption, we purify the noisy EPR pairs in this group as
\begin{equation}
	\sum_{i,j',i'}\sqrt{\frac{q^{jj'}_{ii'}}{q^j}} \mathbb{I} \otimes X_B^{e^j_i}Z_B^{e^{j'}_{i'}} \ket{\psi_{00}^{\otimes n}}_{AB}\ket{e^j_i e^{j'}_{i'}}_E,
\end{equation}
where the notation is given in Table \ref{tab:Notation}. As mentioned in Sec.~\ref{sec:one-way-sec}, Alice adds noise under the constraint that won't affect the bit error syndrome. Recall that $s^j = He^j_i$ and $HG = 0$. Thus, the bit error pattern after adding noise should be in the form of 
\begin{equation}
	\tilde{e}^j_i = e^j_i + Gf, f\in \{0,1\}^k.
\end{equation}
Alice independently samples each bit of $f$  with a probability $p$ of being 1. The probability of sampling each $f$ is $p_f = p^{|f|}(1-p)^{k-|f|}$. Alice applies a bit flip $X_i$ if $(Gf)_i = 1$. After adding noise, the state becomes:
\begin{equation}
\begin{split}
	&\sum_{j,j',i',f} \sqrt{\frac{q^{jj'}_{ii'}}{q^j}p_f} X_A^{Gf} \otimes X_B^{e^j_i}Z_B^{e^{j'}_{i'}} \ket{\psi_{00}^{\otimes n}}_{AB} \ket{f}_{A'}\ket{e^j_i e^{j'}_{i'}}_E \\
	=&\sum_{j,j',i',f}\sqrt{\frac{q^{jj'}_{ii'}}{q^j}p_f} \mathbb{I} \otimes X_B^{e^j_i}Z_B^{e^{j'}_{i'}}X_B^{Gf} \ket{\psi_{00}^{\otimes n}}_{AB} \ket{f}_{A'}\ket{e^j_i e^{j'}_{i'}}_E\\
	=&\sum_{j,j',i',f}\sqrt{\frac{q^{jj'}_{ii'}}{q^j}p_f} \mathbb{I} \otimes X_B^{e^j_i + Gf}Z_B^{e^{j'}_{i'}}\ket{\psi_{00}^{\otimes n}}_{AB} Z^{G^{T}e^{j'}_{i'}}\ket{f}_{A'}\ket{e^j_i e^{j'}_{i'}}_E\\
	=&\sum_{j,j',i',f}\sqrt{\frac{q^{jj'}_{ii'}}{q^j}p_f} \mathbb{I} \otimes X_B^{e^j_i + Gf}Z_B^{e^{j'}_{i'}}\ket{\psi_{00}^{\otimes n}}_{AB} Z^{s^{j'}}\ket{f}_{A'}\ket{e^j_i e^{j'}_{i'}}_E.\\
\end{split}
\end{equation}
The first equality comes from $X \otimes X \ket{\psi_{00}} = \ket{\psi_{00}}$. The second equality comes from commuting $Z_B^{e^{j'}_{i'}}$ and $X_B^{Gf}$ produce $(-1)^{(e^{j'}_{i'})^TGf}$, and this is equivalent to apply $Z^{G^T e^{j'}_{i'}}$ on $\ket{f}$. The third equality comes from $G^T e^{j'}_{i'} = s^{j'}$. 

The consumption of bit error correction will change due to the change of the distribution on error pattern $e^{j}_i$. Denote $j_f$ as the index such that $e^{j_f}_i = e^j_i + Gf$. The probability of getting $\tilde{e}^{j}_i$ is
\begin{equation}
	\tilde{q}^{j}_{i} = \sum_{f} p_f q^{j_f}_{i}.
\end{equation}
The consumption of bit error correction becomes
\begin{equation}
	\tilde{I}_e^j = h(\{\frac{\tilde{q}^{j}_i}{q^j}\}_i).
\end{equation}

Now, define $\ket{\psi} = \sqrt{(1-p)}\ket{0} + \sqrt{p} \ket{1}$ and $\ket{\psi^{j'}} = Z^{s^{j'}} \ket{\psi}^{\otimes n}$. After bit error correction, the state becomes
\begin{equation}
\begin{split}
	&\sum_{j,j',i',f}\sqrt{\frac{q^{jj'}_{ii'}}{q^j}p_f} \mathbb{I} \otimes Z_B^{e^{j'}_{i'}}\ket{\psi_{00}^{\otimes n}}_{AB} Z^{s^{j'}}\ket{f}_{A'}\ket{e^j_i + Gf}_{B'}\ket{e^j_i e^{j'}_{i'}}_E \\
	=&C_{A',B'}\sum_{j,j',i',f}\sqrt{\frac{q^{jj'}_{ii'}}{q^j}p_f} \mathbb{I} \otimes Z_B^{e^{j'}_{i'}}\ket{\psi_{00}^{\otimes n}}_{AB} Z^{s^{j'}}\ket{f}_{A'}\ket{e^j_i}_{B'}\ket{e^j_i e^{j'}_{i'}}_E \\
	=&C_{A',B'}\sum_{j,j',i'}\sqrt{\frac{q^{jj'}_{ii'}}{q^j}} \mathbb{I} \otimes Z_B^{e^{j'}_{i'}}\ket{\psi_{00}^{\otimes n}}_{AB} Z^{s^{j'}}\ket{\psi}_{A'}^{\otimes n}\ket{e^j_i}_{B'}\ket{e^j_i e^{j'}_{i'}}_E \\
	=& C_{A',B'}\sum_{j,j',i'}\sqrt\frac{q^{jj'}_{ii'}}{q^j} Z_A^{e^{j'}_{i'}} \otimes  \mathbb{I}\ket{\psi_{00}^{\otimes n}}_{AB} \ket{\psi^{j'}}_{A'}\ket{e^j_i}_{B'}\ket{e^j_i e^{j'}_{i'}}_E,
\end{split}
\end{equation}
in which $C_{A'B'}\ket{f}_{A'}\ket{e^j_i}_{B'} = \ket{f}_{A'}\ket{e^j_i+Gf}_{B'}$. 

Here comes the key observation of adding noise: If Alice corrects $Z^{e^{j'}_{i'}}$, then the state on subsystem $ABA'B'$ takes the form
\begin{equation}
	C_{A'B'} [\psi_{00}^{\otimes n} \otimes \rho_{A'B'}] C_{A'B'}^{T},
\end{equation}
which is a private state that can generate secure keys \cite{PhysRevLett.94.160502}. Alice holds $\ket{\psi^{j'}}$ which contains information of $s^{j'}$. Thus, she can save on the consumption of phase error correction. If this saving surpasses the consumption of the increment in bit error correction, then adding noise can increase the key rate \cite{PhysRevLett.98.020502}.

Now, we consider the consumption of phase error correction. After identifying $e^j_i$, we must identify phase error pattern $e^{j'}_{i'}$. In fact, with the help of $\ket{\psi^{j'}}$, the saving is $h(\sigma^j_i)$ \cite{PhysRevA.54.1869, PhysRevLett.98.020502}, in which
\begin{equation}
\begin{split}
	\sigma^j_i = \sum_{j'} \frac{q^{jj'}_i}{q^{j}_i} \psi^{j'}.
\end{split}
\end{equation} 

Thus, the phase error consumption is 
\begin{equation}
	I_{p,i}^{j} = h(\{\frac{q^{j,j'}_{i,i'}}{q^j_i}\}_{j',i'}) - h(\sigma^j_i).
\end{equation}

According to Eq.~\eqref{eq:ph_consumption_noOTP}, if OTP is not performed in the two-way communication step, the consumption is
\begin{equation}
	\tilde{I}_{p,i}^{j} = n-k + h(\{\frac{q^{j,j'}_{i}}{q^j_i}\}_{j'}) - h(\sigma^j_i).
\end{equation}

Thus, the key rate formula for the groups with bit error syndrome $s^j$ is
\begin{equation}
\begin{split}
	r^j &= \max(1 - \frac{1}{n}h(\{\frac{\tilde{q}^j_i}{q^j}\}_i) -  \frac{1}{n}\sum_{i} \frac{q^j_i}{q^j}[h(\{\frac{q^{j,j'}_{i,i'}}{q^j_i}\}_{j',i'}) - h(\sigma^j_i)], 0) - \frac{n-k}{n}\quad \text{(with OTP)}, \\
	r^j &= \max(\frac{k}{n} - \frac{1}{n}h(\{\frac{\tilde{q}^j_i}{q^j}_i\}) -  \frac{1}{n}\sum_{i} \frac{q^j_i}{q^j}[h(\{\frac{q^{j,j'}_{i}}{q^j_i}\}_{j'}) - h(\sigma^j_i)],0)\quad \text{(without OTP)}.
\end{split}
\end{equation}
Note that the key rate of the scheme without OTP is always higher than with OTP. The final key rate is
\begin{equation}
	r = \sum_j q^j r^j.
\end{equation}

\bibliography{bibAKD.bib}

\begin{thebibliography}{54}
\bibitem [{non()}]{nonlinear_footnote}%
  \BibitemOpen
  \href@noop {} {}\bibinfo {note} {Some non-linear operations, such as the
  CCNOT (Fig.~\ref{fig:CCNOT}), are also CROs. Thus, there is a possibility
  that non-linear codes can be incorporated into our framework to further
  improve the key rate and the allowable bit error threshold.}\BibitemShut
  {Stop}%
\bibitem [{one()}]{oneway_footnote}%
  \BibitemOpen
  \href@noop {} {}\bibinfo {note} {Lemma \ref{lem:phase_without_OTP} also
  applies to the information reconciliation step. However, as the information
  reconciliation corrects the bit error pattern in all locations, the
  corresponding hashing is global. We conjecture that there is little room for
  improving the key rate by refraining from using OTP in this
  step.}\BibitemShut {Stop}%
\end{thebibliography}

\end{document}